\documentclass{article}
\usepackage[top=1in, bottom=1in,left=1in,right=1in]{geometry}
\usepackage[english]{babel}
\usepackage[utf8]{inputenc}
\usepackage{amsmath}
\usepackage{amsbsy}
\usepackage{float}
\usepackage{authblk}
\usepackage{dirtytalk}
\usepackage{natbib}
\usepackage{mathrsfs}
\usepackage{xcolor}
\usepackage{epigraph}
\usepackage{listings}
\usepackage{subcaption}
\usepackage[export]{adjustbox}
\usepackage{hyperref}
\hypersetup{
    colorlinks=true,
    linkcolor=blue,
    filecolor=magenta,      
    urlcolor=red,
    citecolor=brown,
    pdftitle={The dependence structure of high-frequency stock returns},
    pdfauthor={Aleksy Leeuwenkamp, Wentao Hu},
    pdfstartpage=1
    }
\usepackage[nottoc,numbib]{tocbibind}
\usepackage{amssymb}
\usepackage{amsthm}
\usepackage{bbm}
\usepackage{verbatim}
\usepackage[mathscr]{euscript}
\usepackage{graphicx}
\usepackage[colorinlistoftodos]{todonotes}
\usepackage{url}

\numberwithin{mycor}{subsection}
\newtheorem{myth}{Theorem}
\numberwithin{myth}{subsection}
\newtheorem{myprop}{Proposition}
\numberwithin{myprop}{subsection}

\numberwithin{mylemma}{subsection}
\theoremstyle{definition}
\newtheorem{mydef}{Definition}
\numberwithin{mydef}{subsection}

\numberwithin{mymod}{subsection}

\numberwithin{myrem}{subsection}

\numberwithin{myex}{subsection}
\DeclareMathOperator{\cov}{Cov}
\DeclareMathOperator{\var}{Var}

\DeclareMathOperator{\CoV}{CoVaR}
\DeclareMathOperator{\DCov}{\Delta\text{-CoVaR}}
\DeclareMathOperator{\VaR}{VaR}
\DeclareMathOperator{\Real}{\mathbb{R}}
\providecommand{\keywords}[1]{\textbf{\textit{Keywords---}} #1}
\providecommand{\classificationjel}[1]{\textbf{\textit{JEL---}} #1}
\providecommand{\classificationams}[1]{\textbf{\textit{AMS 2020 subject Classification---}} #1}
\title{New general dependence measures: construction, estimation and application to high-frequency stock returns\thanks{The authors would like to thank Jan Dhaene ,Roman Goncharenko, Hamza Hanbali, Florian Hoffmann and Gertjan Verdickt}}
\author[1]{Wentao Hu\footnote{Email: wentao.hu@kuleuven.be}}
\author[1]{Aleksy Leeuwenkamp\footnote{Corresponding author, Email: aleksy.leeuwenkamp@kuleuven.be}}

\affil[1]{KU Leuven Faculty of Economics and Business (FEB) Department of Accounting, Finance \& Insurance (AFI): Naamsestraat 69 B-3000 Leuven, Belgium}   

\date{\today}

\begin{document}
\renewcommand{\bibname}{References}
\maketitle
\begin{abstract}
    \noindent We propose a set of dependence measures that are non-linear, local, invariant to a wide range of transformations on the marginals, can show tail and risk asymmetries, are always well-defined, are easy to estimate and can be used on any dataset. We propose a nonparametric estimator and prove its consistency and asymptotic normality. Thereby we significantly improve on existing (extreme) dependence measures used in asset pricing and statistics. To show practical utility, we use these measures on high-frequency stock return data around market distress events such as the 2010 Flash Crash and during the GFC. Contrary to ubiquitously used correlations we find that our measures clearly show tail asymmetry, non-linearity, lack of diversification and endogenous buildup of risks present during these distress events. Additionally, our measures anticipate large (joint) losses during the Flash Crash while also anticipating the bounce back and flagging the subsequent market fragility. Our findings have implications for risk management, portfolio construction and hedging at any frequency.
\end{abstract}
\newline
\classificationjel{C58,G01,G11,G12}\newline
\classificationams{62G20,62G30,62H12,62H20,62Pxx,91G45,91G70}
\newpage

\section{Introduction}\label{Sec:intro}
Correlations between financial asset returns are a key element of not just asset pricing but also portfolio management models. Correlations quantify the amount on linear comovement between assets and hence are used to identify underlying sources of comovement but also used to construct portfolios that are well-diversified. In this paper we address the shortcomings of correlations by proposing a set of more general dependence measures which, despite their generality, are intuitive and easy to estimate on any dataset. We apply this measure to explore and show economically important empirical properties of asset return data. \newline \indent Despite their success empirical asset pricing models such as the CAPM \cite{sharpe1964capital,lintnervaluation,fama2004capital} and theoretical asset pricing models such as modern portfolio theory \cite{markowitz,elton1997modern} rely heavily on correlations and these methods exhibit flaws and oversights. This fact becomes particularly clear during times of financial market stress when seemingly well-diversified portfolios suddenly become poorly or entirely undiversified as assets that normally move relatively independently or even in opposite directions all move in one direction\footnote{ Some papers that investigate this phenomenon: \cite{LonginSolnikExtremeCorr,Samoroddependence,asymmcorrang,forbes2002no,extremevaluedpendence,hartmann2004asset,embrechts_2010,van2016systematicdownbeta,CopulaGarchCovar,bearbeta,Mcrash}}. \newline \indent This tail asymmetry also called the leverage effect in financial asset returns is well-known in the literature \cite{cont2001empirical,asymmcorrang} and can be modelled with asymmetric or local correlations \cite{asymmcorrang,EngleAsymmcorr,corrriskportchoice,TJOSTHEIM201333}. \newline \indent However, this still does not address the issue of non-linearities in asset returns during financial crises\footnote{Even outside financial crises well-known theoretical asset pricing models pose that returns are non-linear functions of state variables (\cite{campbell1999force,bansal2004risks,brunnermeier2009market,brunnermeier2014macroeconomic}). Furthermore, \cite{pohl2018higher} show that linear approximations in asset pricing can be woefully inaccurate.} nor does it allow to show the direction of the risks. In other words, correlations still measure linear dependence only. This means correlations can be distorted by non-linear transformation of the asset returns/prices. Furthermore, correlations do not allow to distinguish if the link from asset A to B or from asset B to asset A is stronger. Therefore, they are symmetric to risk while in financial markets risks are known to be rather asymmetric \cite{LonginSolnikExtremeCorr,extremevaluedpendence,portchoicecontagionHawkes,CopulaGarchCovar}. Lastly, since correlations rely on (co)variances one has to be sure that the second moment of the population exists otherwise correlations are meaningless \footnote{The existence of second moments of asset returns is a contentious question within asset pricing with many papers on the topic, some of which are: \cite{fama1963mandelbrot,cont2001empirical,gabaix2007unified,gabaix2009power,taleb2009finiteness,grabchak2010financial}. This makes it hard to justify measures like coskewness or cokurtosis \cite{coskew,coskewcokurt} which require even higher moments to exist. }.\newline \indent Becasue these issues are not just relevant to asset pricing or finance but to statistics in general in this paper we develop a set of dependence measures based on copulas, the $\DCov$  risk measure and the link between the two established by \cite{BERNARDI20178,Jaworski+2017+1+19}. The first set of dependence measures called $\iota$ purely measures the dependence. Hence, it can be seen as a generalized correlation. The second set called $\delta$ measures dependence and its effect on the variable of interest. Hence, it can be seen as a generalized regression beta. Both $\iota$ and $\delta$ come in pairs as within a pair one can examine the lower left and the upper right joint tails separately. \\

\noindent In the finance \footnote{ A selection of papers on this: \cite{LonginSolnikExtremeCorr,asymmcorrang,extremevaluedpendence,ang2006downside, TJOSTHEIM201333,van2016systematicdownbeta,bearbeta,CMIEL202055} } and statistics \footnote{ An overview of commonly used methods can be found in: \cite{embrechts1997modelling,joe1997multivariate,gijbels2021specification}. Some famous measures are: Spearmans rank correlation \cite{spearman1987proof}, Ginis gamma \cite{gini1914ammontare}, Kendalls tau \cite{kendall1938new} and Blomqvists beta \cite{blomqvist1950measure}. } literature there exist many papers on extreme or tail dependence measures, local correlations, local betas and more classical dependence measures like Spearmans Rho or Kendalls Tau. All of these methods either still fall prey to the above mentioned issues, have poor estimation properties or lack any economic interpretation or clear intuition.\newline \indent Recently, in asset pricing the Model-Free Implied Dependence (MFID) was proposed \cite{inghelbrecht2022model}. The MFID is based on the HIX (Herd Behavior Index) by \cite{DHAENE2012357} and measures how comonotonic or positively (non)-linearly dependent asset markets are. However, the MFID relies on an asset pricing model, still requires the existence of variances, only measures global and positive dependence, does not take into account the mentioned asymmetries and requires index option data to be able to estimate it. Despite these drawbacks, the MFID has recently been successfully applied in standard empirical asset pricing models to show that investors really seem to price non-linear dependence risk independent of known risk factors \cite{inghelbrecht2022model}.  We have been able to find only one dependence measure, also based on copulas, that comes close to fixing all of these issues but it is less intuitive \cite{CMIEL202055}. \newline\indent Considering the success of the MFID in asset pricing we propose dependence measures that allows one to measure comonotonicity or the opposite countermonotonicity, allows for the risk and tail asymmetry, are always well-defined, measure local dependence, are well behaved statistically and are relatively easy to estimate and understand. \\

\noindent To demonstrate the practical utility of our dependence measures and address a relatively unaddressed question of what the dependence structure of this particular data looks like, we estimate them using high-frequency (more frequent than daily) stock return data during market distress events such as the Flash Crash of 2010 and the Great Financial Crisis (GFC) of 2008. \newline \indent The statistical properties of financial asset returns have been known to get more complex at higher sampling frequencies \cite{cont2001empirical,grabchak2010financial}. However, these works mainly focus on the univariate properties. Regarding dependence we only know of a few papers that have examined the dependence structure of intraday financial asset returns \cite{Samoroddependence,embrechts2003dependence,ait2016cojumps}. These papers already find that joint distribution of asset returns is non-Gaussian at higher frequencies. \cite{Samoroddependence} provides theoretical and empirical evidence that the dependence properties of high frequency returns persist at lower frequencies. Because these properties require a large sample to estimate reliably high-frequency data are better suited. The persistence in turn implies that findings obtained with high-frequency data do generalize to lower frequencies.  However, this literature has not considered local dependence nor has it investigated asymmetries. Furthermore, all the standard models used in the high-frequency literature assume jump diffusion processes which are linearly coupled across assets \cite{ait2014high,ait2016cojumps}. Consequently, these models impose stricter assumptions on the dependence structure than the data might allow for. \newline \indent Therefore, there exists a clear gap in the literature as it is unknown if this dependence structure is adequate to model the true dependence structure of high-frequency stock returns. Lastly, due to the persistence of the dependence structure under temporal aggregation our findings have implications for managing assets at any frequency. \\

\noindent Our main contribution is the construction a new set of dependence measures that allow us to measure dependence between random variables and the effect on dependence under very general assumptions on the underlying random variables and their dependence structure. \newline \indent Theoretically, we prove a wide range of properties of these dependence measures. These properties make them more general than existing approaches. We provide a nonparametric estimator that does not impose any stringent assumptions. Furthermore, we prove consistency and asymptotic normality of our estimator.\newline \indent Empirically, we find that high-frequency stock returns around market distress events exhibit clear tail asymmetry and lack of diversification not captured by correlations. During some important days of the GFC we even find that extreme profits are stronger comoving than extreme losses. Furthermore, in some cases our measures are even able to anticipate large joint losses and profits in asset markets due to assets starting to comove strongly prior to these events. Lastly, some of these effects clearly persist in lower frequency data. Because high-frequency stock returns data and our method reveal these features while showing where standard models fail our findings show the underlying asset pricing processes are rather complex. Consequently, there are clear and strong implications for portfolio management, risk management and hedging at any frequency.\\

\noindent The paper is structured as follows: Section \ref{Sec:theory} is concerned with the mathematical and statistical theory behind our dependence measures, their construction and their properties. Section \ref{Sec:empirics} deals with estimation of our measures and show their performance on a test case. Section \ref{Sec:data} explains the data used. Section \ref{Sec:res} discusses the results and their implications. Section \ref{Sec:conc} concludes the paper while all the proofs and additional information can be found in the Appendix (Section \ref{Sec:appendix}).

\section{Theory}\label{Sec:theory}
\subsection{Copula}\label{Sec:copdep}
To model the dependence structure between asset prices we use copulas. Copulas are widely used in finance, risk management and engineering to model dependence between random variables \cite{LonginSolnikExtremeCorr,extremevaluedpendence,genest2007everything,embrechts2009,Genest2009,mcneil2015quantitative,Mcrash}.\\ 

\noindent  A copula is a function that links the marginal distributions. Consider a bivariate random vector $(X,Y)$ of asset payoffs with respective marginal distribution functions $F_X(x)$ and $F_Y(y)$.\cite{sklar1959fonctions} shows that the joint distribution function admits the following decomposition:
\begin{equation}
F_{X,Y}(x,y) = \mathbb{P}(X\leqslant x, Y\leqslant y) =C(F_X(x),F_Y(y)),   
\end{equation}
where $C(u,v) = \mathbb{P}(U\leqslant u, V\leqslant v)$ and $U=F_X(X),V=F_Y(Y) \sim U(0,1)$. Hence, the copula treats the marginals as uniformly distributed. The function $C$ is called the copula of $(X,Y)$ \footnote{Under the assumption that both marginal distributions are continuous the copula is unique.}. More details about copulas can be found in \cite{nelsen2007introduction,hofert2019elements}.\newline \indent A useful property of a copula is that is invariant to strictly increasing functions of its marginals. Also it does not depend on its marginals but rather shows how the marginals are linked. Hence, a given copula can be used to link a wide variety of marginals to construct a joint distribution. Furthermore, the dependence between the marginals can be non-linear and exhibit features like tail dependence and asymmetries.\newline \indent Another convenient property of copulas is that an upper and lower bound exist on the set of all copulas. Hence, there exist notions of maximal positive and maximal negative dependence. These bounds are given by the Theorem of Fréchet-Hoeffding \cite{frechet1951tableaux,hoeffding1940masstabinvariante}:

\begin{myth}\label{Th:FH}
    Let $(X,Y)$ be a bivariate random vector with copula $C(u,v)$. Then it holds that
    \begin{align*}
        W(u,v)\leqslant C(u,v) \leqslant M(u,v)
    \end{align*}
    \noindent for all $(u,v)\in [0,1]^2 $ with $W(u,v)=\max(u+v-1,0)$ and $M(u,v)=\min(u,v)$.
\end{myth}

\noindent 
If $C(u,v)=M(u,v)$ the random variables are said to be comonotonic which implies that $Y=T(X)$ almost surely for some increasing function $T(\cdot)$. 
This implies there exists just one underlying source of randomness and all marginals move perfectly up- or downwards with this source \cite{dhaene2002concept,DHAENE2012357,inghelbrecht2022model}.\newline \indent If $C(u,v)=W(u,v)$ the random variables are said to be countermonotonic which implies that $Y=T(X)$ almost surely for some decreasing function $T(\cdot)$ (see \cite{nelsen2007introduction}). In two dimensions one can see this as perfect comovement but in the opposite direction. \newline \indent Because $T(\cdot)$ is allowed to be a non-linear function, copulas capture also non-linear dependence. Therefore, comonotonicity neatly formalizes the intuition that in times of distress financial asset returns strongly comove (\cite{asymmcorrang,ang2006downside,kelly2014tail,van2016systematicdownbeta,bearbeta,Mcrash,inghelbrecht2022model}). The Fréchet-Hoeffding result is what makes copulas a particularly useful tool to model dependence structures as it allows to go beyond the linear dependence structure of the bivariate Gaussian distribution in an elegant manner. The result by Fréchet-Hoeffding can be expanded by splitting the set of copulas into positively dependent, negative dependent and independent copulas. \\

\noindent To do so we first define positive and negative quadrant dependence. Let $(X,Y)$ be a bivariate random vector with copula $C(u,v)$. $(X,Y)$ is positive (negative) quadrant dependent  or PQD(NQD) if for all $(x,y)\in\Real^2$ we have that
    \begin{align*}
        \mathbb{P}(X\leqslant x, Y\leqslant y)\geqslant(\leqslant) \mathbb{P}(X\leqslant x)\mathbb{P}(Y\leqslant y),
    \end{align*}
or in terms of the copula for all $(u,v)\in[0,1]^2$ we have that
    \begin{align*}
        C(u,v)\geqslant(\leqslant) uv.
    \end{align*}


\begin{myth}\label{TH:FH2}
     Let $(X,Y)$ be a bivariate random vector with copula $C(u,v)$. If $X$ and $Y$ are positive quadrant dependent (PQD) then it holds that
    \begin{align*}
        I(u,v)\leqslant C(u,v) \leqslant M(u,v).
    \end{align*}
    \noindent If $X$ and $Y$ are negative quadrant dependent (NQD) then it holds that
    \begin{align*}
        I(u,v)\geqslant C(u,v) \geqslant W(u,v).
    \end{align*}
    \noindent for all $(u,v)\in [0,1]^2 $ with $W(u,v)=\max(u+v-1,0)$, $M(u,v)=\min(u,v)$ and $I(u,v)=uv$ the independence copula.
\end{myth}

\noindent 
The result and definitions for PQD and NQD can be found in \cite{nelsen2007introduction} for example. This result alongside results on our dependence measures being an increasing function of the copula (dependence consistency results) are crucial for the dependence measures to work consistently and be well-defined.

\subsection{CoVaR and $\Delta$-CoVaR}\label{Sec:CovDcov}

To quantify the marginal risk and the dependence between marginal risks we use the VaR, CoVaR and $\Delta$-CoVaR. We denote a payoff by $X$ and its cumulative distribution function by $F_X(x)$.The VaR of $X$ at probability level $\alpha \in [0,1]$ is defined as
    \begin{equation}
        \VaR_{\alpha}(X) = F^{-1}_X(\alpha), 
    \end{equation}
    with $F^{-1}_X(\alpha)$ the quantile function of $X$. While the VaR gives a threshold payoff under a certain probability it cannot describe the losses which exceed this threshold. However, since the VaR is based on the quantile function it always exists. To measure the dependence between the payoffs of assets the CoVaR was proposed.  For given loss probabilities $\alpha,\beta \in (0,1)$, \begin{equation}\label{1ver}
        \CoV_{\alpha,\beta}(Y| X)=\VaR_{\beta}(Y| X \leqslant \VaR_{\alpha}(X)).
\end{equation}
An often used definition of the CoVaR from \cite{AdrianBrunnCovar} is given as
\begin{equation}\label{2ver}
    \CoV^{=}_{\alpha,\beta}(Y| X)=\VaR_{\beta}(Y| X= \VaR_\alpha(X)).
\end{equation}

\noindent In this paper, we use definition (\ref{1ver}). \cite{MainikSchaanning+2014+49+77} provides a comprehensive comparison of the definition (\ref{1ver}) and (\ref{2ver}). They prove that $\CoV^{=}$ is not an increasing function of the dependence structure and hence it is not dependence consistent whereas $\CoV$ is. Since we require dependence measures that are increasing functions of the underlying dependence structure this definition of the $\CoV$ is used. Next, one can measure the risk contribution of $X$ to $Y$ by comparing $\CoV_{\alpha,\beta}(Y| X)$ to the $\CoV$ in some benchmark state. In the literature the following definitions exist. Given $\alpha,\beta \in (0,1)$, we can define
\begin{align} 
        \DCov_{\alpha,\beta}(Y| X)&=\CoV_{\alpha,\beta}(Y| X)-\VaR_\beta(Y),\\
        \DCov^{=,\text{med}}_{\alpha,\beta}(Y| X)&=\CoV^=_{\alpha,\beta}(Y| X)-\CoV^=_{1/2,\beta}(Y| X). \label{2ex}
    \end{align}

\noindent The $\DCov^{=,\text{med}}$ was first proposed by \cite{AdrianBrunnCovar}. In general, the $\DCov_{\alpha,\beta}(Y| X)$ can be interpreted as a risk contribution measure which is the change in $VaR_{\beta}(Y)$ if $VaR(X)$ moves from a benchmark state to a distressed state. In the first definition the benchmark state $\VaR_\beta(Y)$ can be seen as the $\CoV$ of $Y$ given $X$ under independence. This allows for an easier statistical interpretation as a dependence measure because it allows to isolate what risk is due to the dependence between the random variables instead of their marginal risk. Furthermore, in \cite{dhaene2022systemic} it is proven that this definition is dependence consistent under more general assumptions. Therefore, we will use this definition of the $\DCov$. 

\newpage

\subsection{Copulas and the CoVaR}\label{Sec:copCov}

Since the $\CoV$ involves the conditional distribution of random variables there exists a link between copula and the $\CoV$. This link was first establised by \cite{MainikSchaanning+2014+49+77} for $\CoV$, by \cite{hakwa2015analysing} for $\CoV^=$ and further elaborated upon in \cite{BERNARDI20178} for both versions. Following \cite{BERNARDI20178} we use their $\CoV$ definition related to the joint lower left tail (joint loss) exceedance: 

\begin{mydef}\label{def:CovCopula}
    Let $(X,Y)$ be a bivariate random vector with copula $C(u,v)$. Then the lower left $\CoV_{\alpha,\beta}$ is defined as follows:
    \begin{align*}
        F_{Y| X\leqslant \VaR_\alpha(X)}=\mathbb{P}(Y\leqslant \CoV_{\alpha,\beta}| X\leqslant \VaR_\alpha(X))=\frac{C(\alpha,F_Y(\CoV_{\alpha,\beta}))}{\alpha}=\beta.
    \end{align*}
    \noindent Let $F_Y(\CoV)=\omega$. Then, at a significance level $\beta\in(0,1)$ $\omega(\alpha,\beta,C)$ is the largest solution \footnote{Using any simple root finding algorithm like Newton-Raphson the solution can easily be found numerically. Because continuity is satisfied by all copulas and copulas are increasing in both arguments the solution is often unique too \cite{nelsen2007introduction}. Also, analytical solutions exist for copulas shown in \cite{BERNARDI20178}.} to the following equation:
    \begin{align*}
        C(\alpha,\omega)=\alpha\beta
    \end{align*}
    \noindent and therefore $\CoV_{\alpha,\beta}(Y| X)= F^{-1}_Y(\omega)=\VaR_\omega (Y)$.
\end{mydef}
\noindent 
This definition makes apparent that the full information of the dependence structure is contained within $\omega$. Also, the beyond the level $\omega$ the $\CoV$ just depends on the marginal distribution of $Y$. To obtain $\CoV$ under exceedances appropriate transformations of the copula need to be used. For example, in this paper we are also interested in the $\CoV$ related to the joint upper right (joint profit) exceedance. In this case $\CoV$ is defined as 

\begin{mydef}\label{def:Covcopula2}
Let $(X,Y)$ be a bivariate random vector with copula $C(u,v)$. Then the upper right $\CoV_{\alpha,\beta}$ is defined as follows:
    \begin{align*}
    \mathbb{P}(Y\geqslant \CoV_{\alpha,\beta}| X\geqslant \VaR_\alpha(X))=\frac{\Bar{C}(\alpha,F_Y(\CoV))}{1-\alpha}=1-\beta.
\end{align*}
\noindent Where the joint tail function is $\Bar{C}(u,v)=1-u-v+C(u,v)$. Let $F_Y(\CoV_{\alpha,\beta})=\gamma$. Then $\gamma(\alpha,\beta,C)$ is the smallest solution to  $\Bar{C}(\alpha,\gamma)=(1-\alpha)(1-\beta)$ and $\CoV_{\alpha,\beta}(Y| X)=F^{-1}_Y(\gamma)=\VaR_{\gamma}(Y)$.
\end{mydef}

\noindent 
In the following sections unless otherwise specified the lower left tail CoVaR from definition \ref{def:CovCopula} is used. It can easily be seen that all results on the lower left tail CoVaR in this paper equally hold for the right upper tail CoVaR with different bounds and flipped inequalities. Two more definitions for the upper left and lower right tail CoVaR could be analogously defined. However, in finance one is mostly interested in the comovement of large losses and profits. Hence, the provided definitions suffice for this paper. Using the Fréchet-Hoeffding bounds and dependence consistency \cite{BERNARDI20178,Jaworski+2017+1+19} establish that:
\begin{itemize}
    \item $\VaR_{\alpha\beta}(Y)\leq\CoV_{\alpha,\beta}(Y\mid X)\leq \VaR_{\beta}(Y)$ because $\alpha\beta\leq \omega\leq \beta$ for PQD payoffs
    \item $\VaR_{1-\alpha(1-\beta)}(Y)\geq\CoV_{\alpha,\beta}(Y\mid X)\geq \VaR_{\beta}(Y)$ because $1-\alpha(1-\beta)\geq\omega\geq\beta$ for NQD payoffs.
\end{itemize} \noindent The lower bound in the former is the $\CoV/\omega$ under comonotonicity and the upper bound in the latter is the $\CoV/\omega$ under countermonotonocity. Therefore, by using these bounds we will construct dependence measures using $\omega$ and the $\DCov$ which will work similar to a classical Pearson correlation. However, in contrast to a classical Pearson correlation our dependence measures are non-linear, local, invariant under a wider range of transformations on the marginals and always well-defined \footnote{Problems with the Pearson correlation are posed in \cite{embrechts2001correlation,hofert2019elements}.}. Furthermore, computing our measures is possible on any given dataset and does not require options data like \cite{DHAENE2012357,bearbeta,inghelbrecht2022model}. 
\newpage

\subsection{Dependence measures}\label{Sec:delta}
\subsubsection{Definitions}\label{Ssec:def}
The first set of dependence measures are constructed using $\omega/\gamma$ depending on if we consider lower left tail or upper right tail exceedances. 
\begin{mydef}\label{def:iota}
    Let $(X,Y)$ be a random vector then the proposed dependence measures $\iota^L_{\alpha,\beta}(Y\mid X),\iota^U_{\alpha,\beta}(Y \mid X)$ are defined as follows:
\end{mydef}
    \begin{equation}
        \iota^L_{\alpha,\beta}(Y\mid X)=    \begin{cases}
        \frac{\omega-\beta}{\alpha\beta-\beta} & \text{if }\omega-\beta \leqslant 0 \\
        \\
        \frac{-(\omega-\beta)}{(1-\alpha)(1-\beta)} & \text{if } \omega-\beta> 0
    \end{cases}
    \end{equation}
        \begin{equation}
        \iota^U_{\alpha,\beta}(Y\mid X)=    \begin{cases}
        \frac{\gamma-\beta}{\alpha-\alpha\beta} & \text{if }\gamma-\beta> 0 \\
        \\
        \frac{\gamma-\beta}{\alpha\beta} & \text{if } \gamma-\beta\leqslant 0
    \end{cases}
    \end{equation}
\noindent With $\omega,\gamma$ defined in definition \ref{def:CovCopula} and \ref{def:Covcopula2} respectively.\\
\noindent The big upshot of $\iota$  is that it satisfies the properties of a dependence measure as given in \cite{nelsen2007introduction} while being more general by measuring dependence locally at the quantiles $\alpha,\beta$. All other relevant properties will be examined in the next section. There, we show that the $\iota$ measure works similarly to a correlation albeit more generalized. This measure has a rather straightforward interpretation. Lets say $\alpha=\beta=0.05$. Then, for example $\iota_{0.05,0.05}^L(Y\mid X)$ measures to what extent and in which direction $Y$ below its 5\% level comoves with $X$ given $X$ moves below its 5\% level. Hence, intuitively its evident that $\iota_{0.05,0.05}^L(Y\mid X)$ might not be equal to $\iota_{0.05,0.05}^L(X\mid Y)$. Therefore, the measure  behaves like a generalized correlation. This interpretation is what sets $\iota$ apart from a very similar measure developed by \cite{CMIEL202055}. Additionally, $\iota$ allows for an easy extension to a measure that also takes into account the risk the dependence poses on $Y$. This measure will be called $\delta$ and is defined as follows:

\begin{mydef}\label{def:delta}
    Let $(X,Y)$ be a random vector then the proposed dependence risk measures $\delta^L_{\alpha,\beta}(Y| X),\delta^U_{\alpha,\beta}(Y| X)$ are defined as follows:
    \begin{equation}
        \delta^L_{\alpha,\beta}(Y| X)=    \begin{cases}
        \frac{\DCov_{\alpha,\beta}(Y| X)}{\VaR_{\alpha\beta}(Y)-\VaR_{\beta}(Y)} & \text{if }\DCov_{\alpha,\beta}(Y| X)\leqslant 0 \\
        \        \frac{-\DCov_{\alpha,\beta}(Y| X)}{\VaR_{1-\alpha(1-\beta)}(Y)-\VaR_{\beta}(Y)} & \text{if } \DCov_{\alpha,\beta}(Y| X)> 0
    \end{cases}
    \end{equation}
        \begin{equation}
        \delta^U_{\alpha,\beta}(Y| X)=    \begin{cases}
        \frac{\DCov_{\alpha,\beta}(Y| X)}{\VaR_{\alpha+\beta-\alpha\beta}(Y)-\VaR_{\beta}(Y)} & \text{if }\DCov_{\alpha,\beta}(Y| X)> 0 \\
        \\
        \frac{-\DCov_{\alpha,\beta}(Y| X)}{\VaR_{\beta(1-\alpha)}(Y)-\VaR_{\beta}(Y)} & \text{if } \DCov_{\alpha,\beta}(Y| X)\leqslant 0
    \end{cases}
    \end{equation}
    \noindent With the $\DCov$ in $\delta^L$ following definition \ref{def:CovCopula} and the $\DCov$ in $\delta^U$ following definition \ref{def:Covcopula2}.
\end{mydef}
\noindent The $\delta$ measure acts more like a regression coefficient because it takes dependence and the effect of dependence on the $Y$ variable into account. By using the quantile function this measure will always be well-defined and it contains information on all moments of $Y$. Also, the interpretation and intuition is rather simple. If we take the example used above then $\delta_{0.05,0.05}^L(Y\mid X)$ not just shows the extent and direction of the comovement of $Y$ below its 5\% level given $X$ is below its 5\% level but also what the effect of this comovement on $Y$ is. Therefore, $\delta$ shows how close the dependence structure and the risk of $Y$ are close to the comonotonic upper bound given a shock in $X$. From the definitions its evident that the measures are strongly linked. We make this link clear by showing that if the distribution of $Y$ is (locally) uniform then the two measures will be equivalent (for proof in the Appendix \ref{proof:unifeq}). This corresponds to the notion that $\delta$ is just like $\iota$ but with information on the distribution $Y$ incorporated in it. Conversely, $\iota$ can be seen as $\delta$ when one is indifferent about or has no information on the distribution of $Y$.  

\subsubsection{Properties}\label{Ssec:prop}
Both measures have some enticing properties which set them apart from other dependence measures \cite{gijbels2021specification}. All proofs can be found in the Appendix (Section \ref{App:proofs}). The properties of $\iota$ are:
\begin{enumerate}
  \item Dependence consistency. 
    \item $\iota\in[-1,1]$, $\iota=-1\Longleftrightarrow$ counter-monotonicity, $\iota=1 \Longleftrightarrow$ comontonicity and $\iota=0 \Longleftrightarrow$ independence. 
    \item Invariant under strictly increasing functions of $X$ and $Y$.
    \item  Risk asymmetry: $\iota_{\alpha,\beta}(Y\mid X)= \iota_{\alpha,\beta}(X\mid Y)$ $\Longleftrightarrow$ the copula is symmetric $C(u,v)=C(v,u)$. 
    \item Tail asymmetry: $\iota^L_{\alpha,\beta}(Y\mid X)=\iota^U_{1-\alpha,1-\beta}(Y\mid X)$ $\Longleftrightarrow$ the copula is radially symmetric $C(u,v)=\Bar{C}(1-u,1-v)$.
    \item Locality: $\iota_{\alpha,\beta}(Y\mid X)$ measures the dependence between $X$ and $Y$ at the quantiles given by $\alpha,\beta$.
    \item Consistency and asymptotic normality.
\end{enumerate}
\noindent These properties hold for both the upper as lower tail version of $\iota$. All properties are proven in the Appendix. Properties 2-6 make $\iota$ more general than the Pearson correlation, Spearmans Rho and Kendalls Tau. In the literature local versions of the Pearson correlation do exist \cite{asymmcorrang,TJOSTHEIM201333} but these only relax on locality and hence measure local linear dependence.\newline \indent Economically, risk asymmetry can be interpreted as comovement in $Y$ given a move in $X$ being more or less likely than comovement in $X$ given a move in $Y$. This can for example be seen in international stock markets where a movement in US stocks is more likely to affect stock markets of other countries than a move in any of these stock markets affecting US stock markets. In distress states this results in asymmetric channels of contagion. This asymmetry is modelled in finance using Hawkes processes \cite{portchoicecontagionHawkes}. Tail asymmetry is also sometimes referred to as the leverage effect and it means in asset markets that during a large downward moves asset returns are stronger correlated than during equally large upward moves \cite{LonginSolnikExtremeCorr,asymmcorrang}. In finance this phenomenon is often modelled with asymmetric correlations or asymmetric jumps \cite{EngleAsymmcorr,corrriskportchoice,ait2009portfoliojumpschoice,portchoicecontagionHawkes}.\newline \indent Since the $\delta$ measure works more like an asset pricing beta than a correlation it has some different properties. The properties of $\delta$ are the following:
\begin{enumerate}
\item Dependence consistency. 
    \item If $F_Y^{-1}$ is strictly increasing: $\delta\in[-1,1]$, $\delta=-1\Longleftrightarrow$ counter-monotonicity, $\delta=1 \Longleftrightarrow$ comontonicity and $\delta=0 \Longleftrightarrow$ independence. 
    \item Invariant under strictly increasing functions of $X$ and strictly increasing linear functions of $Y$.
    \item  Risk asymmetry: $\delta_{\alpha,\beta}(Y\mid X)= \delta_{\alpha,\beta}(X\mid Y)$ $\Longleftarrow$ the copula is symmetric $C(u,v)=C(v,u)$ and $X$ is a linear function of $Y$. 
    \item Tail asymmetry: $\delta^L_{\alpha,\beta}(Y\mid X)=\delta^U_{1-\alpha,1-\beta}(Y\mid X)$  $\Longleftarrow$ the copula is radially symmetric $C(u,v)=\Bar{C}(1-u,1-v)$ and $Y$ is symmetric.
    \item Locality: $\delta_{\alpha,\beta}(Y\mid X)$ measures the dependence between $X$ and $Y$ at the quantiles given by $\alpha,\beta$.
\item Consistency.
    \end{enumerate}
\noindent Again, the properties hold for both the upper right and lower left tail version of $\delta$. The main difference is by taking the distribution of $Y$ into account through the quantile function. Therefore, it acts similarly to an asset pricing beta. This comes at the cost of property 7 and the strength of properties 2-5. However, again it is more general than a beta by measuring non-linear and local dependence. Furthermore, using property 2 one can show $\delta$ is invariant to location-scale shifts of $Y$ and $X$. This means that drift and volatility will not affect $\delta$ which is useful as asset pricing data often exhibits a drift and volatility clustering \cite{cont2001empirical}. Lastly, these measures always exist, can be computed on any data set and are quite simple\footnote{The HIX comonotonicity measure requires index option prices and the existence of the variance \cite{DHAENE2012357}. }.

\section{Methodology}\label{Sec:empirics}
\subsection{Estimation}
In order to estimate $\iota$ and $\delta$ there are many possible techniques. One must estimate the copula in some way to obtain $\omega$ or $\gamma$ and then one can estimate the Value-at-Risk of $Y$ using a multitude of methods. However, one must mind that using certain techniques may impose restrictions that imply restrictions on the properties of these measures. For example, most bivariate GARCH models impose a normal or t-copula with symmetric margins. In this case all the special properties of the measures become irrelevant as the statistical model only allows for linear dependence and has one global dependence structure. Hence, ideally one must estimate these measures as distribution free and data driven as possible. In \cite{leeuwenkamp2022making} a data driven estimator for $\omega/\gamma$ and the $\DCov$ was developed we will use this estimator in this paper too.\\

\noindent First, in accordance with the copula estimation literature such as \cite{hofert2019elements} the observations are first transformed to the interval $[0,1]$ with the empirical quantile function. Then, on these pseudo-observations the empirical beta copula (see Appendix \ref{Cop:empbeta} for more details) estimator proposed by \cite{SEGERS201735} is used. The empirical beta copula is a fully nonparametric copula estimator with desirable properties such as continuity \footnote{ According to the Intermediate Value Theorem and its corollary Bolzano's Theorem at least continuity is required to guarantee the root of a function and hence the solution in terms of $\omega$ to $C(\alpha,\omega)-\alpha\beta=0$ exists. } and improved small sample performance in terms of mean squared error compared to other nonparametric copula estimators while being asymptotically equivalent to the empirical copula estimator \cite{SEGERS201735}. Also, it exhibits no boundary bias and the smoothing is fully determined by the sample size \footnote{On the other hand, kernel based estimators suffer from boundary bias and require a bandwidth parameter \cite{omelka2009improved}.}. Once an estimate $\hat{C}(u,v)$ is obtained using this method then also an estimate $\hat{\Bar{C}}(u,v)$ can be easily obtained. Then, by solving the equations
\begin{align*}
    \hat{C}(\alpha,\omega)-\alpha\beta&=0
\\ \hat{\Bar{C}}(\alpha,\gamma)-(1-\alpha)(1-\beta)&=0\end{align*}
\noindent in terms of $\omega,\gamma$ estimates for $\omega$ and $\gamma$ are obtained. One can see here that to compute the estimate of $\omega$ or $\gamma$ at a certain pair of quantiles $\alpha,\beta$ one just needs to estimate the entire copula once. This in contrast to local correlations and betas where one needs to re-estimate the entire dependence structure at different thresholds.\newline \indent Using $\hat{\omega}$ or $\hat{\gamma}$ $\iota^L$ and $\iota^U$ are easy to estimate. To estimate $\delta$ one additionally needs to estimate the VaR of $Y$. There is an extensive literature on how to estimate univariate VaRs \cite{Condvarcompari}. However, to keep the method fully empirical we estimate the VaR of $Y$ using the empirical quantile function.\newline \indent In the Appendix it will be proven that estimates of $\omega/\gamma$ and hence $\iota^L/\iota^U$ are consistent and asymptotically normal in a point-wise and uniformly. As far as the authors are aware it is not possible to prove asymptotic normality for $\delta$ estimates as the asymptotic distribution is a ratio of normals . By imposing some functional form on $F^{-1}$ like for example a standard normal or t quantile function convergence properties could be improved. However, this comes at the cost of imposing assumptions such as symmetry on the margins and a possible increased asymptotic (co)variance \cite{genest2010covariance}.\newline \indent Because the estimators are nonparametric assumptions like linear dependence, symmetry of the copula, symmetry of the margins , radial symmetry of the copula, tail (in)dependence and the existence of moments are not imposed. Hence, the estimators are always well-defined if one assumes continuous random variables. Lastly, due to the properties in Section \ref{Ssec:prop} estimates can be used to verify features like risk (a)symmetry, tail (a)symmetry and non-linear dependence in the data.\\

\noindent To estimate the measures in a time-dependent fashion in this paper we opt for expanding window estimation because it keeps the data to estimate fixed apart from the observations that are added at each time step. While it does require assuming stationarity within each window we aim to clearly show the time-dependent dependence structure through the changes between each window. Lastly, for the sake of simplicity and computational efficiency the expanding window method was chosen. 
\newpage
\subsection{Test Case}\label{ssec:case}
To show that our dependence measures alleviate some of the problems with correlations and work similar to the MFID/HIX while being data independent we use a test case from \cite{embrechts2001correlation,DHAENE2012357} to show that $\iota$ and $\delta$ remain consistent whereas the correlation fails. The test case has been adapted for our purposes and is defined as follows
\begin{enumerate}
    \item Generate 5000 observations $(X_1,X_2)$ from a bivariate normal distribution with $\rho=0.95, \mu_1=0,\mu_2=0,\sigma_1=0.2, \sigma_2=s$.
    \item Take $S_i=\exp(X_i), i=1,2$. 
    \item Estimate the empirical beta copula on $(S_1,S_2)$, estimate $\iota^L_{0.5,0.5}(X_2\mid X_1),\iota^U_{0.5,0.5}(X_2\mid X_1),\delta^L_{0.5,0.5}(X_2\mid X_1),\delta^U_{0.5,0.5}(X_2\mid X_1), \rho(S_1,S_2)$ and compute the theoretical values. 
    \item Vary $s$ from 0.2 to 2 in steps of 0.01 and repeat steps 1-3 for each $s$.
    \item Plot for each $s$ the $\iota/\delta$ estimates, $\rho$ estimates and the theoretical $\iota^L,\iota^U,\delta^L,\delta^U,\rho(S_1,S_2)$. 
\end{enumerate}
\noindent Because the $X_i$ are normally distributed $S_i=\exp(X_i)$ will have a lognormal distribution with $\mathbb{E}[S_i]=\exp(\mu_i), \var(S_i)=\exp(2\mu_i)(\exp(\sigma^2_i)-1),i=1,2$. The correlation is then
\begin{align*}
    \rho(S_1,S_2)=\frac{\exp(\rho(X_1,X_2)\sigma_1\sigma_2)-1}{\sqrt{\exp(\sigma^2_1)-1}\sqrt{\exp(\sigma^2_2)-1}},
\end{align*}
\noindent which is a decreasing function in $\sigma_2$. The exponential transformation will not affect the copula: a Gaussian copula but now with lognormal margins. According to property 3 of both $\iota$ and $\delta$ an increasing transformation of $X_1$ will not change $\delta$ and $\iota$. However, now $X_2$ is also transformed with a non-linear increasing transformation which can change $\delta$ but won't change $\iota$. To assess robustness of $\delta$ in this case the theoretical values of $\delta$ are computed from the Gaussian copula with normal margins (equivalent to a bivariate normal distribution) and compared to $\delta$ estimates obtained from the empirical beta copula.\newline \indent A significance level of $\alpha=\beta=0.5$ has been chosen as then the $\iota$'s and $\delta$'s cover the entire distribution from center to the tails and are therefore most similar to $\rho$.\newline \indent The test case scenario is meant to emulate an event of stress in the market where two assets are strongly comoving but where the volatility of one asset relative to another increases. The scenario could be applicable for an individual stock that crashes versus an index that contains said stock. 
\begin{figure}[H]
     \centering
     \begin{subfigure}[t]{0.38\textwidth}
    \includegraphics[width=\textwidth]{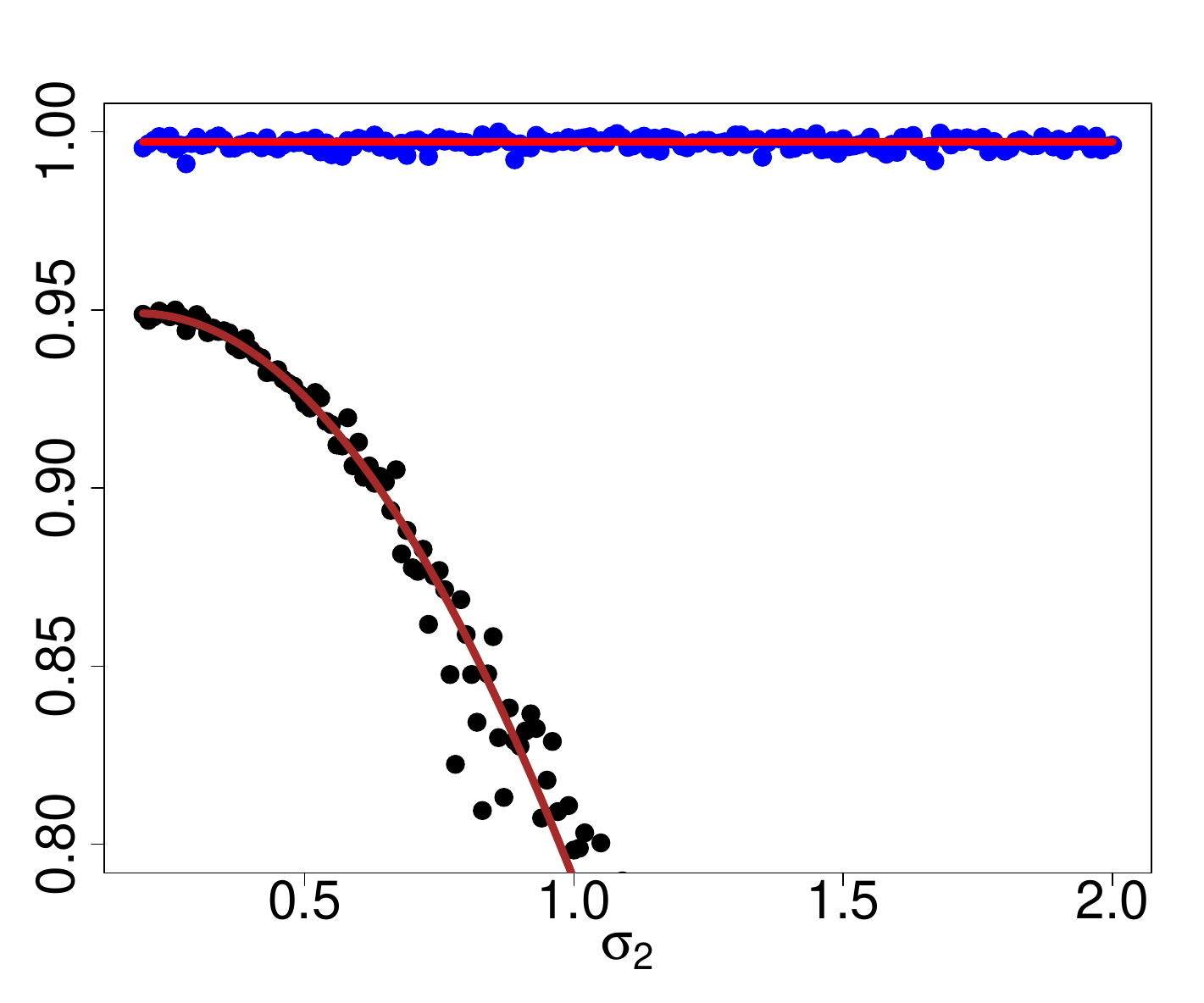}
    \caption{}
    \label{fig:Simplotiota}
\end{subfigure}
\begin{subfigure}[t]{0.38\textwidth}
    \includegraphics[width=\textwidth]{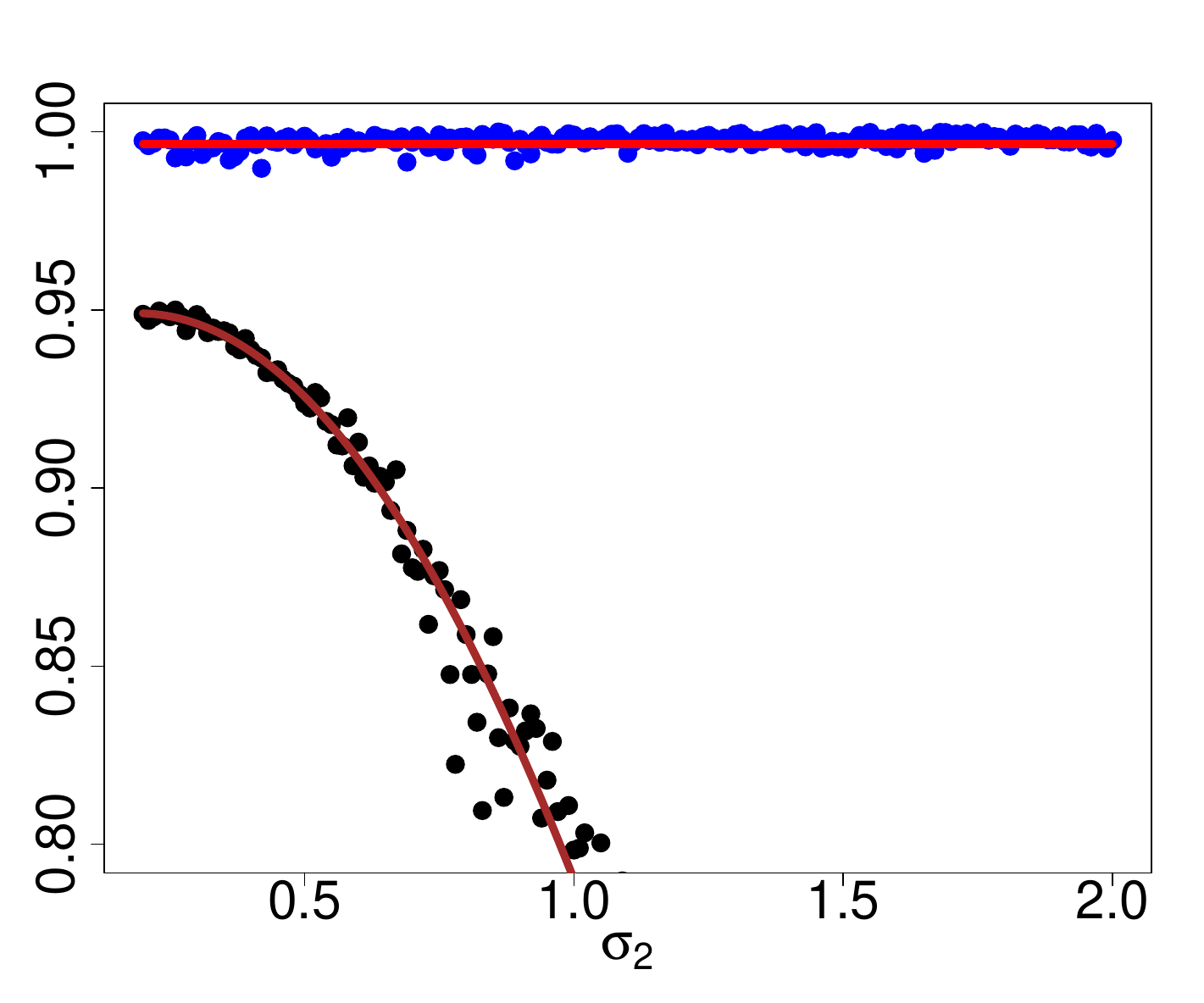}
    \caption{ }
    \label{fig:Simplotdelta}
\end{subfigure}
\caption{Estimates of $\iota^L_{0.5,0.5}(Y\mid X)$ (left) and $\delta^L_{0.5,0.5}(Y\mid X)$ (right) in the blue dots versus the theoretical values (red line) compared to the correlation estimates (black dots) and their theoretical values (brown line). All are in function of $\sigma_2$. }
\label{fig:comps}
\end{figure}
\noindent In Figure \ref{fig:comps} it is clear how sensitive correlations are to this transformation of the data while $\iota$ and $\delta$ estimates remain consistent in showing the underlying unchanged dependence structure. Because the Gaussian copula is symmetric and the margins are a linear function of each other it only suffices to show only one of each (lower left tail in this case). The figure shows a striking resemblance to Figure 1 in \cite{DHAENE2012357} which shows indeed just like the HIX $\iota$ and $\delta$ measure non-linear dependence.\newline \indent Furthermore, $\iota$ estimates due to their invariance under strictly increasing transformations remain truly consistent to the theoretical values from the untransformed distribution. On the other hand, the $\delta$ estimates show a slight deviation as they seem to increase. This is verified by means of linear regression of the $\delta$ estimates on $\sigma_2$ which shows a significant positive coefficient estimate whereas that of the $\iota$ estimates on $\sigma_2$ is insignificant. Therefore, the effect of the non-linear transformation of $X_2$ seems to have a negligible effect on the $\delta$ estimates here. \newline \indent Lastly, the test case is also a comforting confirmation of statistical properties as $\iota$ estimates seem to exhibit no bias and have a tightly controlled variance. The $\delta$ estimates exhibiting similar behavior is encouraging too despite the lack of a formal proof.  
\section{Data}\label{Sec:data}

For the empirical results we use the data from \cite{pelgerHF2020}\footnote{ Obtained from the authors website: \href{https://mpelger.people.stanford.edu/data-and-code}{https://mpelger.people.stanford.edu/data-and-code} }. The dataset contains TAQ (Trade-and-Quote) data of S\&P 500 stock returns from January 2nd 2004 to December 30th 2016 aggregated at the 5 minute level. To ensure enough trades to compute returns the data each day starts at 9:35 instead of 9:30. This results in 77 returns observations per stock per day. The dataset is a balanced panel of 332 stocks to ensure it consists of liquid stocks and aid estimation. For more details on the dataset we refer the reader to \cite{pelgerHF2020}. High frequency data is used because it provides large sample sizes and is found to have similar statistical properties to lower frequency data \cite{Samoroddependence}.  \newline\indent
To reduce the impact of microstructure effects on estimation the 5 minute aggregation is chosen in accordance with the high-frequency econometrics literature \cite{ait2014high}\footnote{To avoid the effect of bid-ask bounces the volume weighted transaction price at the last second of each 5 minute interval was used.}. For robustness and to investigate temporal aggregation effects we also apply our methods to returns aggregated at the 10 minute, 30 minute and 1 hour levels. \newline \indent
Since the dataset is rather large in the cross-sectional and the time-series dimension our main interest lies in the dependence structure during market distress events, we only estimate our measures in an expanding window with data of 30 prior days\footnote{ This implies $30\times77=2310$ observations at least per stock. To have this many observations with daily data one would need more than 9 years of data assuming 250 trading days/year.} on the following days: September 15th 2008 (Lehman Bankruptcy), September 16th 2008 (AIG bailout), September 29th 2008 (TARP fails to pass Congress), October 3rd 2008 (amended TARP passed by Congress) and May 6th 2010 (the Flash crash). The GFC days were chosen because prior stress was already built-up in financial markets due to the GFC unfolding with additional events on the respective days. Next, we consider the Flash Crash of May 6th 2010. On this day the Dow Jones Industrial Average (DJIA) reached -7\% in a matter of minutes before quickly recovering to around -2\%. Since the causes of this Flash Crash have been determined to be mostly endogenous \cite{easley2011flash,filimonov2012quantifying,menkveld2019flash}, the speed and severity of the downturn provide a contrast to the GFC. \newline \indent The working hypothesis is if sudden market stress is mostly endogenously driven then our dependence measures computed on individual stock returns conditional on market returns\footnote{ In our case we use the market factor constructed by \cite{pelgerHF2020}. This a value-weighted average of all stocks in CRSP minus the daily risk-free rate.} should pick up increasing dependence of joint losses early. Similar reasoning applies to the recovery and joint profits. In contrast, longer-run stress build-up should result in already high values. Lastly, we suspect our measures on lower frequency returns will still show some of these effects.\newline \indent Therefore, our dependence measures could not just help market makers and high-frequency traders to keep track of intraday risk in their inventories but also market participants working at lower frequencies. Early risk detection can result in significantly reduced losses as \cite{easley2011flash} and \cite{menkveld2019flash} either provide examples or estimates of significant market maker losses on this day.

\newpage
\section{Results}\label{Sec:res}
In all the results that follow we have computed both the lower left tail as well as upper right tail measures of the individual stock returns $Y$ conditional on the stock market returns $X$ at levels $\alpha=\beta\in\{0.01/0.99,0.05/0.95,0.10/0.90,0.5/0.5\}$. Next, the results of each individual stock are averaged to obtain a measure of average dependence/lack of diversification in the stock markets at a particular point in time. Averaging returns before computing the dependence measures just results in a measure always close to 1 as equally weighted S\&P 500 stock and CRSP returns are virtually comonotonic.\newline\indent Both the fully nonparametric as well as a fully parametric (Gaussian copula and margins) estimation approaches are used. The fully Gaussian approach is a baseline meant to demonstrate how different the actual dependence structure is from those assumed in standard or even high-frequency asset pricing models. Hence, it aims to demonstrate that these models do not adequately capture what happens in the markets with potentially costly consequences. \\ \indent

First, due to its uniqueness in US stock market history the results of the Flash Crash are shown.
\begin{figure}[H]
     \centering
     \begin{subfigure}[t]{0.38\textwidth}
    \includegraphics[width=\textwidth]{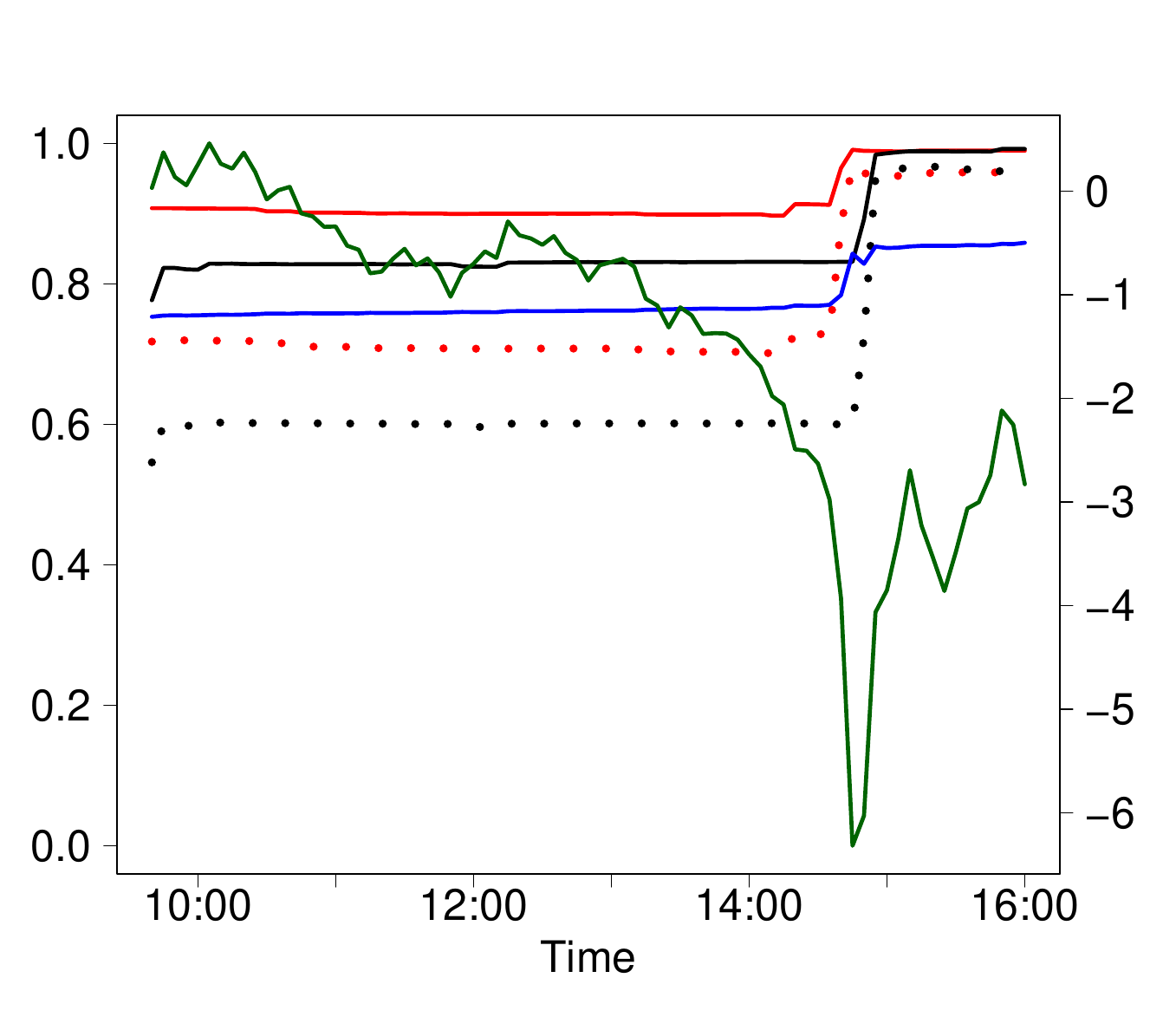}
    \caption{L: $\alpha=\beta=0.01$, U: $\alpha=\beta=0.99$}
    \label{fig:FC99}
\end{subfigure}
\begin{subfigure}[t]{0.38\textwidth}
    \includegraphics[width=\textwidth]{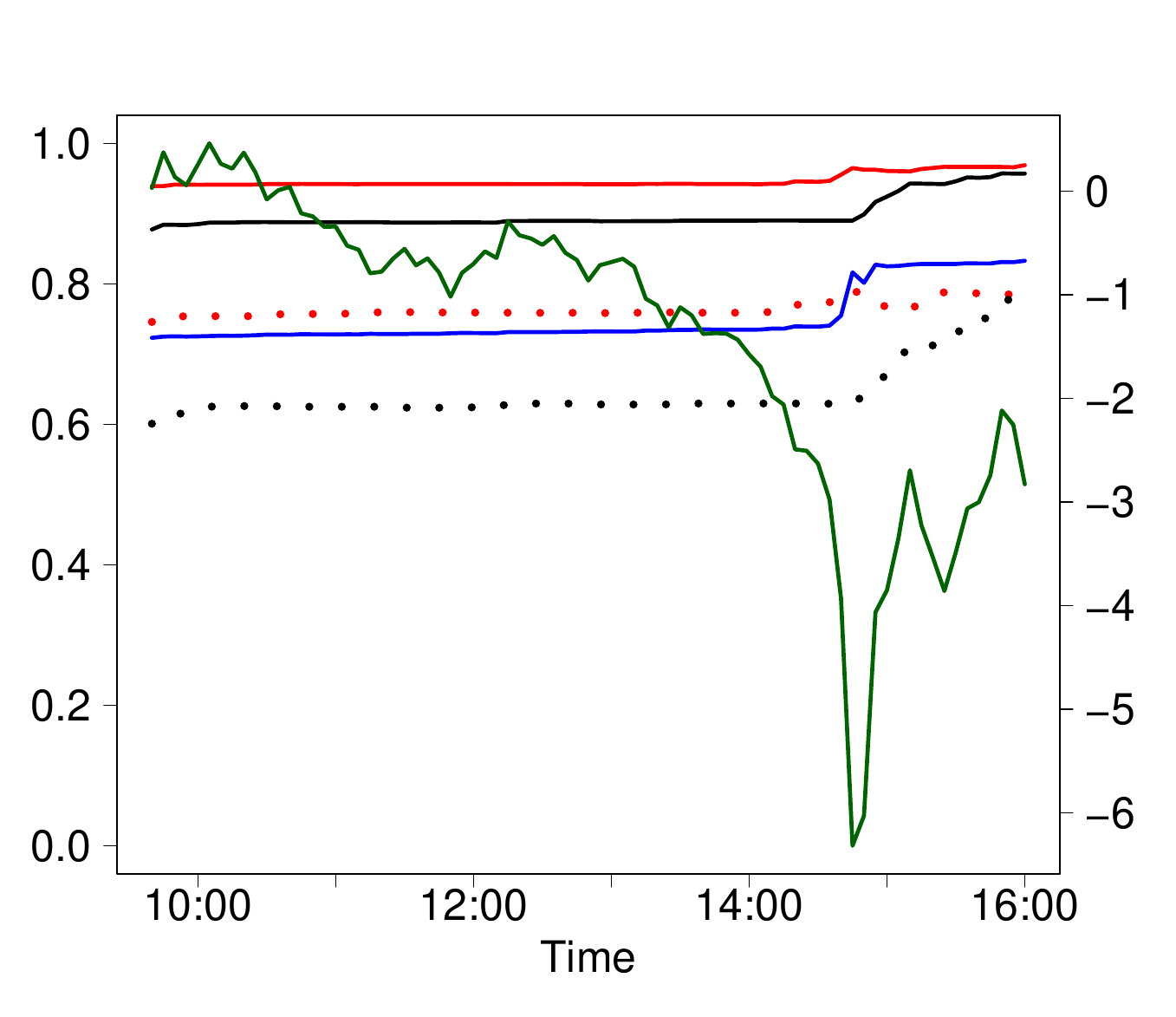}
    \caption{L: $\alpha=\beta=0.05$, U: $\alpha=\beta=0.95$ }
    \label{fig:FC95}
\end{subfigure}
\begin{subfigure}[t]{0.38\textwidth}
    \includegraphics[width=\textwidth]{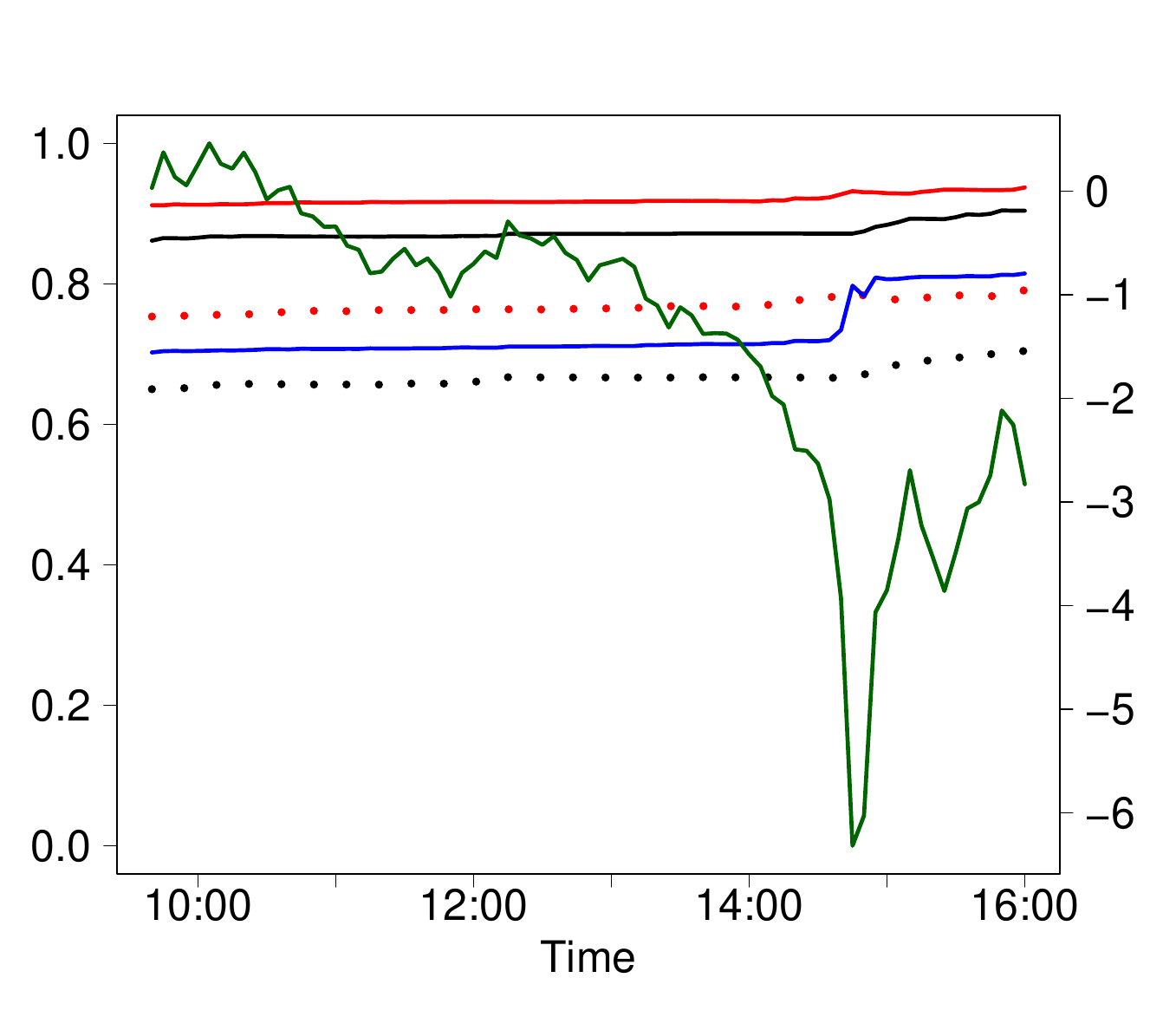}
    \caption{L: $\alpha=\beta=0.10$, U: $\alpha=\beta=0.90$}
    \label{fig:FC90}
\end{subfigure}
\begin{subfigure}[t]{0.38\textwidth}
    \includegraphics[width=\textwidth]{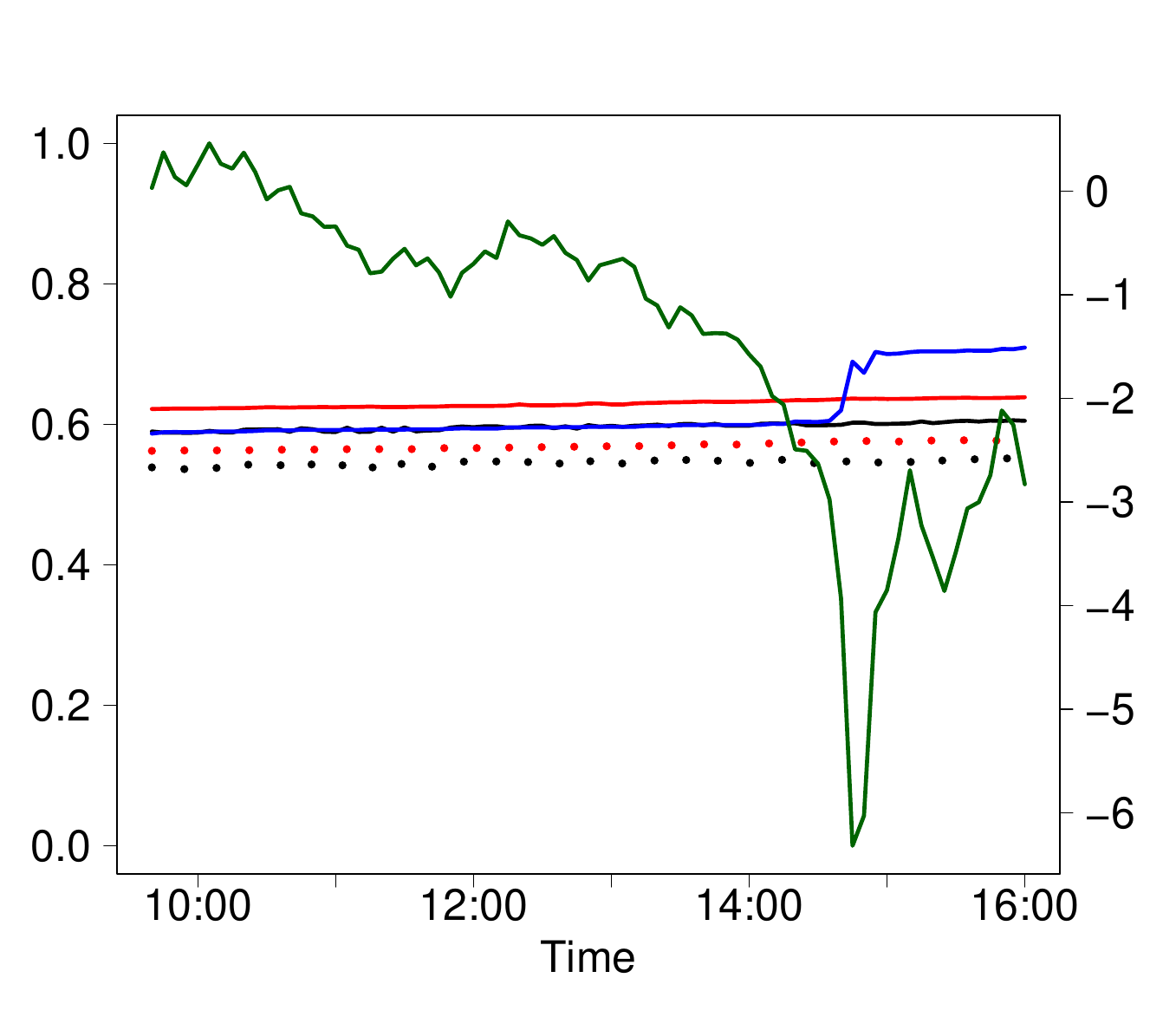}
    \caption{L: $\alpha=\beta=0.50$, U: $\alpha=\beta=0.50$ }
    \label{fig:FC50}
\end{subfigure}
\caption{ Left: Estimates of the lower (red) and upper tail (black) $\iota$(solid) $\delta$(dotted) on the Flash Crash data compared with the fully Gaussian $\delta$(blue solid). Right: Market index \% return (green line)}
\label{fig:compsFC}
\end{figure}
\noindent In Figure \ref{fig:compsFC} some results are immediately clear: the tail asymmetry is striking and most prominent at more extreme quantiles, but disappears after the big drop. This suggests the returns have become more symmetric in the tails with both tails now being strongly dependent. Furthermore, at the $1\%$ quantile the increased cross asset dependence before the big drop in the market is signalled. Rather stunningly, the first jump of the measure occurs already at 14:20, the second jump at 14:40 and the third jump at 14:45. The timing of the first jump is 10 minutes before \cite{menkveld2019flash} define the start of the Flash Crash while the second jump is 5 minutes before the large drop in the market. The third jump occurs at the bottom and 1 minute after cross-market arbitrage fully broke down. These findings persist more subtly at less extreme quantiles and applies to the $99\%$ quantile for large profits as well. \newline \indent These findings are consistent with an endogenous story behind the Flash Crash where deteriorating market conditions caused market makers to pull out, weakening cross-market arbitrage in the process and making the market even more vulnerable to large sell orders \cite{easley2011flash,filimonov2012quantifying,menkveld2019flash}. After the drop some institutions entered the market again thereby instigating the recovery. Also, the measures persist and indicate the potential for large market losses or gains has not disappeared which is also consistent with \cite{easley2011flash,filimonov2012quantifying,menkveld2019flash}. Hence, without any order level data and with using data at a lower frequency than order data we manage to still find the effects of these endogenous mechanisms.\newline\indent In contrast, all of these effects are mostly missed by the fully Gaussian based version of $\delta$ which only indicates an increase in risk but cannot tell in what part of the distribution, it is too late and at any quantile it is just a shifted version of the correlation. Therefore, using models built on correlations in this situation would have flagged the risk too late and leave one ignorant on what part of the distribution the risk manifests. \newline\indent Examining the same data with our measures at frequencies of 10 minutes, 30 minutes and 1 hour we find similar patterns albeit more smoothed out due to the lower sample size (see Appendix \ref{App:figs}). \\

\noindent Next, to compare with news-driven market stress we compute the measures on the GFC dates. 
\begin{figure}[H]
     \centering
     \begin{subfigure}[t]{0.38\textwidth}
    \includegraphics[width=\textwidth]{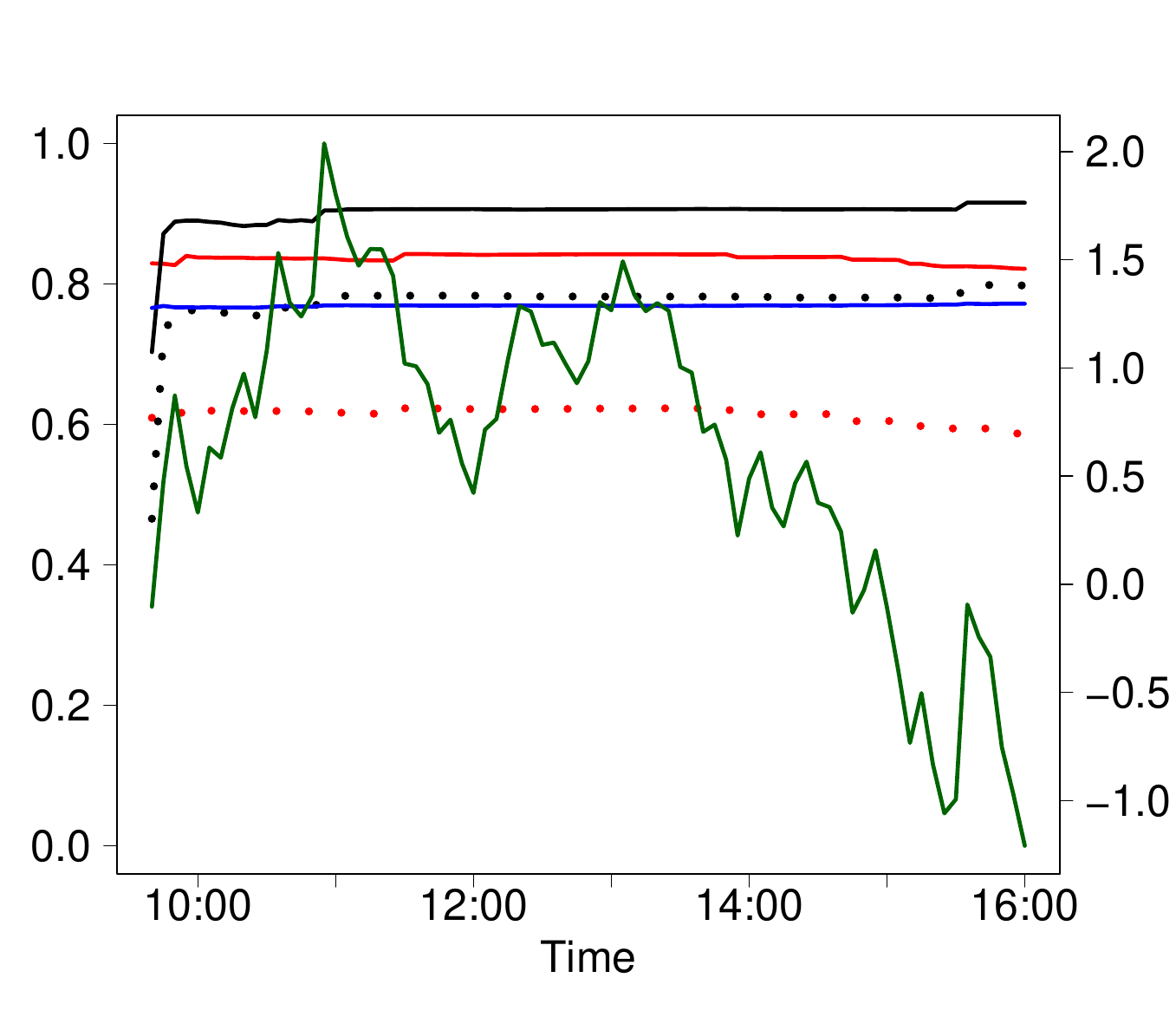}
    \caption{September 15th 2008}
    \label{fig:GFC1}
\end{subfigure}
\begin{subfigure}[t]{0.38\textwidth}
    \includegraphics[width=\textwidth]{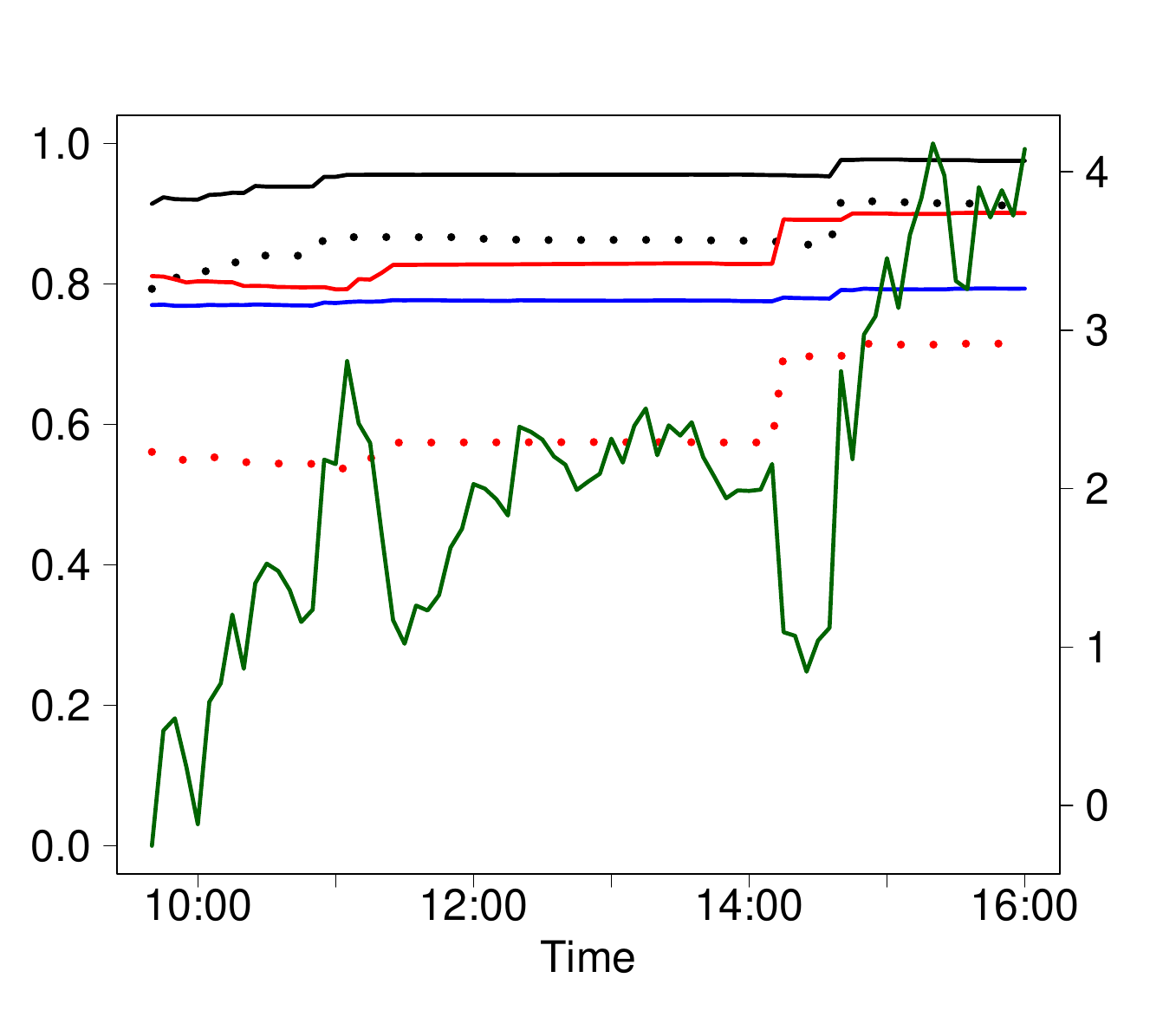}
    \caption{September 16th 2008}
    \label{fig:GFC2}
\end{subfigure}
\begin{subfigure}[t]{0.38\textwidth}
    \includegraphics[width=\textwidth]{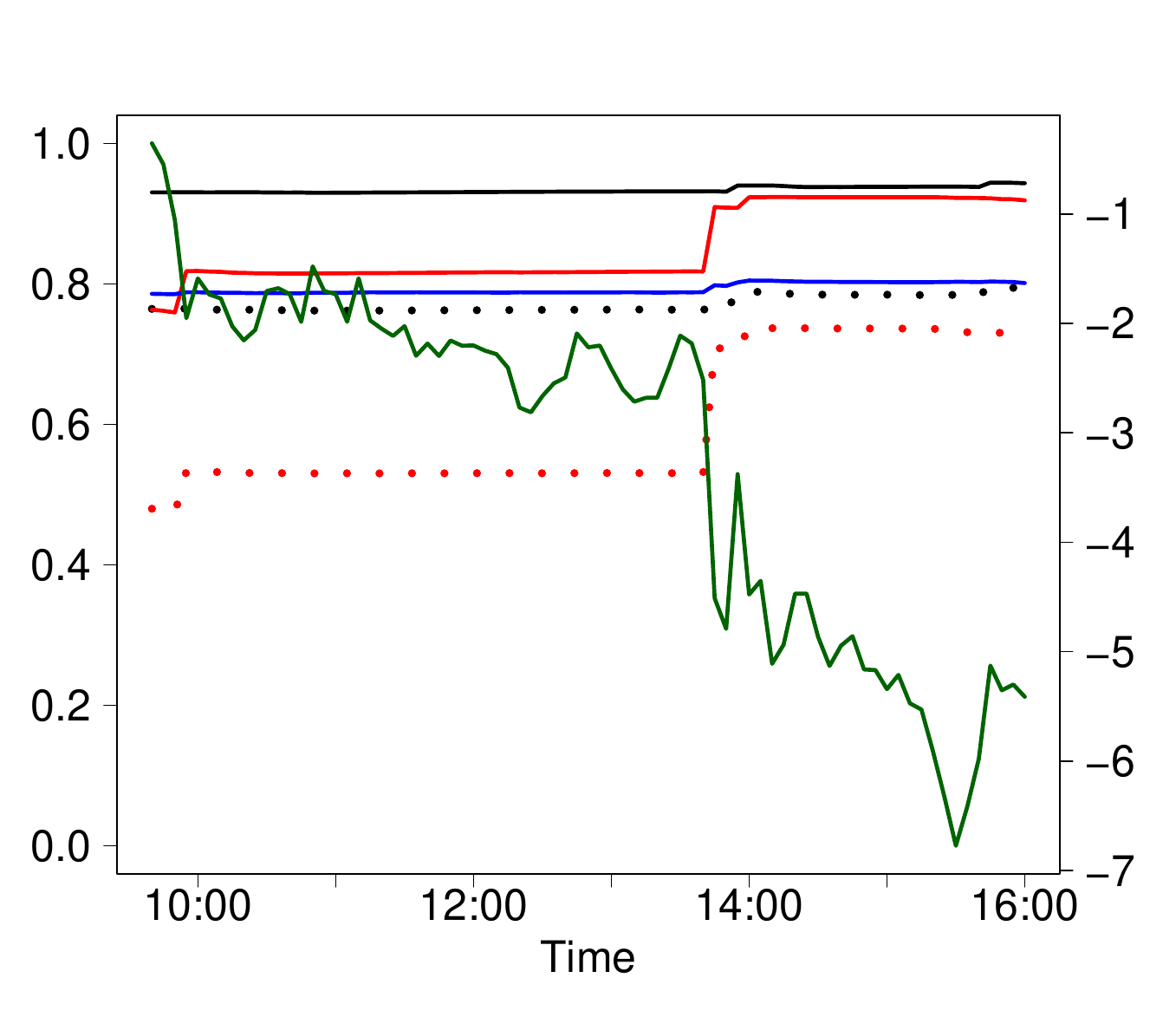}
    \caption{September 29th 2008}
    \label{fig:GFC3}
\end{subfigure}
\begin{subfigure}[t]{0.38\textwidth}
    \includegraphics[width=\textwidth]{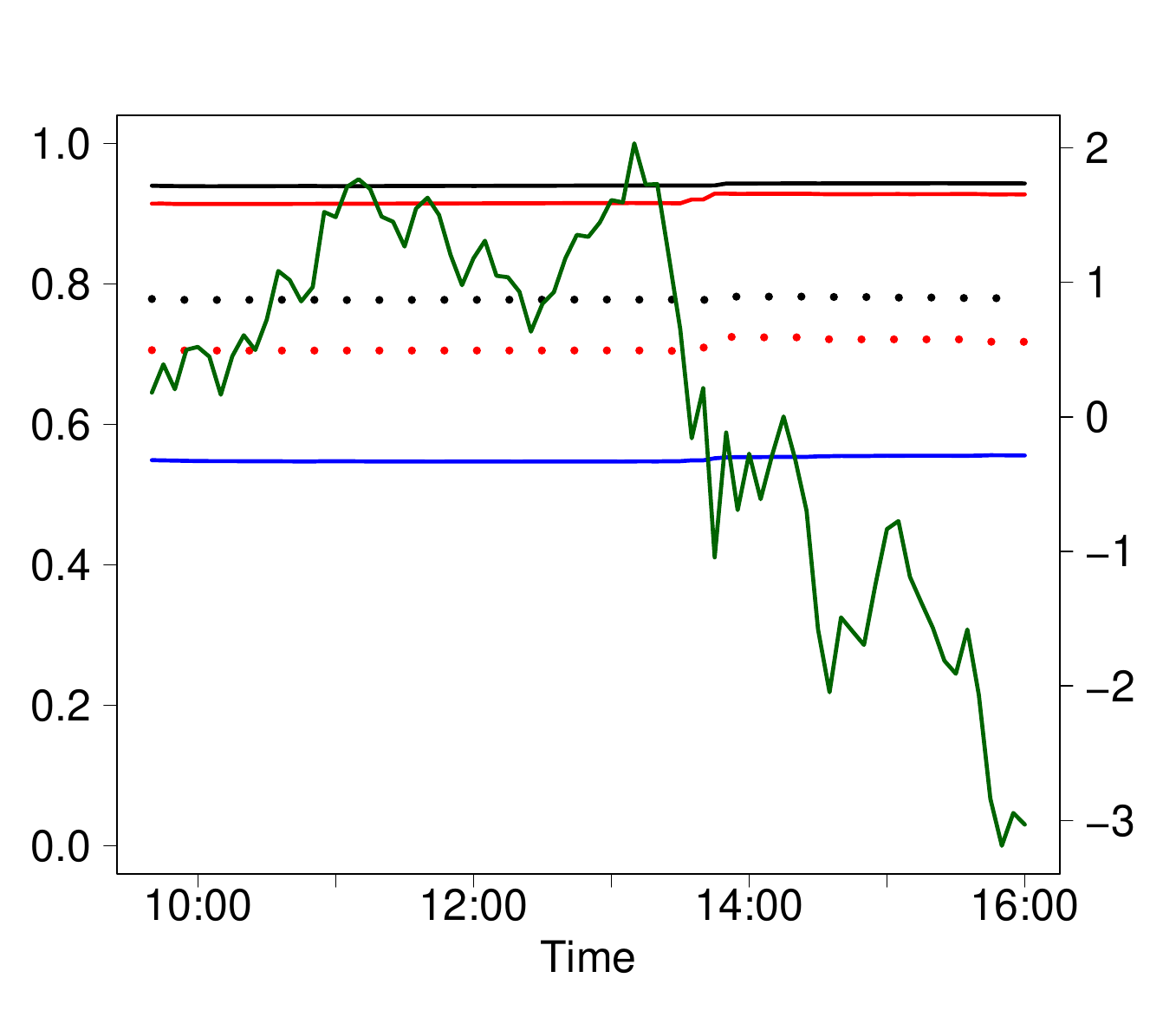}
    \caption{October 3rd 2008}
    \label{fig:GFC4}
\end{subfigure}
\caption{Left: estimates of the lower (red) and upper tail (black) $\iota$(solid) $\delta$(dotted) on the GFC event data compared with the fully Gaussian $\delta$(blue solid).Right: market index \% return (green line). All evaluated at: Lower $\alpha=\beta=0.01$, Upper $\alpha=\beta=0.99$.}
\label{fig:GFC}
\end{figure}
\newpage
In Figure \ref{fig:GFC} it is again evident from the difference between the Gaussian based $\delta$ and the nonparametric estimates that in the tails returns have a dependence structure which differs strongly from that of a Gaussian copula. On some of the days the Gaussian $\delta$ even barely reacts.\newline\indent Even more strikingly, early on September 15th the upper $\iota$ and $\delta$ exceed the lower ones. As an immediate market reaction to the bankruptcy of Lehman Brothers it shows initially markets reacted positively with large positive and strongly comoving returns across S\&P500 stocks. The pattern persists the next day and even on September 29th and October 3rd. This finding is novel as usually comovements for extreme losses are stronger. Again, we confirm that the dependence structure of stock returns over time can be more complex than current high-frequency asset pricing models like those in \cite{ait2016cojumps,pelgerHF2020} allow for. \newline\indent Furthermore, there are some jumps in the measures. In the Appendix we also examine the measures at lower quantiles and confirm the jumps and the extreme profit comovement only seem to occur at sufficiently high quantiles. Consequently, this confirms the notion they are caused by (co)jumps in the asset prices.\newline \indent Across all other quantiles we also the measures are elevated and indicate loss symmetry indicating the very volatile and fragile financial market situation where diversification opportunities are scarce and there is quick succession of extreme (co)movements of assets in both directions. \newline \indent Since the correlation based $\delta$ does not capture these dynamics at all it perfectly shows the losses traders and market makers can incur by basing their hedging strategies and risk management on traditional correlations.

\section{Conclusion}\label{Sec:conc}
\indent We have devised a set of dependence measures that are very general and powerful yet easy to estimate and interpret. We prove the measures have some useful properties that sets them apart from existing dependence measures \cite{gijbels2021specification} and make them particularly suited to applications in finance. However, these measures can be estimated on any sufficiently large dataset and are therefore useful to any scientific field. \newline\indent
We develop fully nonparametric but smooth estimators for our measures. For the estimator of one of the measures we prove both point-wise and uniform consistency and asymptotic normality and show the advantages versus standard Pearson correlations on a test case. \newline \indent
Next, using high-frequency stock return data around well-known financial market distress events we demonstrate that our dependence measure are able to show financial market dynamics and comovement patterns that standard correlations are too late to pick up or unable to pick up at all. The most prominent example being that our measures anticipate both the extreme joint losses and profits present in financial markets during the 2010 Flash Crash. \newline \indent We show these  comovement patterns persist at lower frequencies. Therefore, we can state that a lot of the stylized special facts of asset returns well-known in the literature are present in the data at any frequency. However, one must have a sufficiently large sample size and flexible estimation approach to really see the difference. Consequently, even for lower frequency applications high-frequency data offer more insight.\newline \indent To make this conclusion more intuitive we pose the following analogy with photography\footnote{ Since cameras perform sampling of the incoming light signal this analogy is more direct than one would think: \cite{shahram2004imaging,shahram2006images1statistical}.}: the true underlying asset pricing process is likely quite complex like a detailed landscape. When taking pictures of the landscape or gathering data and estimating a model in asset pricing we necessarily smudge out some of the complexity due to the resolution of the camera or frequency of the data and the structure imposed by the model. A higher resolution camera is able to resolve more of the details of the landscape. Analogously, high-frequency data and more flexible models capture more of the stylized features that are undoubtedly present in asset pricing data while the underlying price process or landscape in all cases is identical. Hence, it is optimal to use the highest frequency data possible even if the timescale of interest is larger. The simple conclusion contained in this analogy echoes that of \cite{ait2005often} but for different reasons.\newline\indent Therefore, by showing correlations can fail to spot lack of diversification, comovements in assets and do miss the devil in the (asymmetric) tails our findings have strong implications for portfolio management, risk management and hedging at any frequency. \newpage
\bibliographystyle{apalike}
\bibliography{ref}
\newpage
\section{Appendix}\label{Sec:appendix}

\subsubsection{Empirical beta copula}\label{Cop:empbeta}
The empirical beta copula is an empirical copula estimator but also a copula itself. For a sample $\boldsymbol{X}_i=(X_{i,1},\dots, X_{i,d}),i=1,\dots,n$ which is ranked $R^{(n)}_{i,j}$ among $X_{1,j},\dots,X_{n,j}$ the empirical beta copula is:
\begin{align*}
    \mathbb{C}^\beta_n(\boldsymbol{u})=\frac{1}{n}\displaystyle\sum^n_{i=1}\prod^d_{j=1}F_{n,R^{(n)}_{i,j}}(u_j) \text{ with }\boldsymbol{u}=(u_1,\dots,u_d)\in[0,1]^d
\end{align*}
\noindent With $F_{n,r}(u)=\displaystyle\sum^n_{s=r}\begin{pmatrix}
    n \\s
\end{pmatrix} u^s(1-u)^{n-s}$ the CDF of a beta distribution $B(r,n+1-r)$. In two dimensions this reduces to a sample $(X_i,Y_i),i=1,\dots,n$ with $R^{(n)}_{i,j},j=1,2$ denoting the ranks of $X$ and $Y$ respectively. Then
\begin{align*}
     \mathbb{C}^\beta_n(u,v)=\frac{1}{n}\displaystyle\sum^n_{i=1}F_{n,R^{(n)}_{i,1}}(u)F_{n,R^{(n)}_{i,2}}(v) \text{ with }(u,v)\in[0,1]^2.
\end{align*}
\noindent The general idea behind the beta copula is that the ranks of a uniformly distributed sampled i.i.d. from the sample $(X_i,Y_i)$ will be beta distributed. However, this is equivalent to sampling from the empirical beta copula conditional on the original sample. This is what allows one to easily sample from this copula. \newline\indent
Furthermore, the empirical beta copula is equivalent to the empirical Bernstein copula where all Bernstein polynomials are of degree $n$ \cite{janssen2012largebernstein,SEGERS201735}. Since in the proof of Weierstrass Approximation Theorem Bernstein polynomials can be used \cite{bernstein1912démo} it means the empirical beta copula estimated on a sample of size $n$ can approximate any copula arbitrarily well as the sample size grows. For finite samples it means that one can approximate well any copula that at most is as complex as an $n$-th degree Bernstein polynomial. This increase in smoothness as the sample size grows does away with the need for any bandwidth parameter common in other smooth estimation approaches of copulas \cite{omelka2009improved}.\newline
\indent Due to the Fundamental Theorem of Algebra Bernstein polynomials also help with the $\omega$ and $\gamma$ estimates as they require solving a root finding problem with the estimated copula. 

\subsection{Proofs}\label{App:proofs}
\subsubsection{Proof of equivalence of $\iota$ and $\delta$}\label{proof:unifeq}

\begin{myprop}
    Let $(X,Y)$ be a bivariate random vector with copula $C$. If $Y\sim U(a,b)$ then $\iota_{\alpha,\beta}(Y\mid X)=\delta_{\alpha,\beta}(Y| X)$ for all $\alpha,\beta\in(0,1)$.
\end{myprop}
\begin{proof}
     Because $Y$ is distributed uniformly on $(a,b)$ $F_Y^{-1}(p)=p(b-a)+a$. Therefore $\iota_{\alpha,\beta}(Y \mid X)=\delta_{\alpha,\beta}(Y| X).$
\end{proof}
\noindent This proposition holds for the lower left and upper right tail versions of $\iota$ and $\delta$. Also, if $Y$ is not uniform but its quantile function is linear on the intervals $[\alpha\beta,\beta]$, $[\beta,1-\alpha(1-\beta)]$, $[\beta,\alpha+\beta-\alpha\beta]$ or $[\beta(1-\alpha),\beta]$ depending on the version used. In economic terms the result states that if one is completely ignorant about the distribution of $Y$ (apart from that it has finite support) then there is no difference in considering $\iota$ or $\delta$. \newpage
\subsubsection{Proofs of properties of $\iota$}\label{proof:iota}
We aim to establish the following properties: 
\begin{enumerate}
    \item Dependence consistency. 
    \item $\iota\in[-1,1]$, $\iota=-1\Longleftrightarrow$ counter-monotonicity, $\iota=1 \Longleftrightarrow$ comontonicity and $\iota=0 \Longleftrightarrow$ independence. 
    \item Invariant under strictly increasing functions of $X$ and $Y$.
    \item  Risk asymmetry: $\iota_{\alpha,\beta}(Y\mid X)= \iota_{\alpha,\beta}(X\mid Y)$ $\Longleftrightarrow$ the copula is symmetric $C(u,v)=C(v,u)$. 
    \item Tail asymmetry: $\iota^L_{\alpha,\beta}(Y\mid X)=\iota^U_{1-\alpha,1-\beta}(Y\mid X)$ $\Longleftrightarrow$ the copula is radially symmetric $C(u,v)=\Bar{C}(1-u,1-v)$.
    \item Locality: $\iota_{\alpha,\beta}(Y\mid X)$ measures the dependence between $X$ and $Y$ at the quantiles given by $\alpha,\beta$.
    \item Consistency and asymptotic normality.
\end{enumerate}
\noindent 1. Dependence consistency:\\
To prove dependence consistency we have to prove that $\iota^L_{\alpha,\beta}(Y\mid X)\leq\iota^L_{\alpha,\beta}(Y'\mid X')$ if $C(u,v)\leq C'(u,v)$ for all $u,v\in(0,1)$ and $F_X=F_{X'},F_Y=F_{Y'}$. 
\begin{proof}
    Take two bivariate random vectors $(X,Y),(X',Y')$ with copulas $C(u,v),C'(u,v)$. Assume that  $C(u,v)\leq C'(u,v)$ for all $u,v\in(0,1)$ and $F_X=F_{X'},F_Y=F_{Y'}$. Then we have at any level $\alpha,\beta\in(0,1)$ that $\omega'\leq \omega$. Therefore, $\iota^L_{\alpha,\beta}(Y\mid X)\leq\iota^L_{\alpha,\beta}(Y'\mid X')$.
\end{proof}
\noindent This proof equally applies for the upper tail version but then with some inequalities flipped. \\

\noindent 2. $\iota\in[-1,1]$, $\iota=-1\Longleftrightarrow$ counter-monotonicity, $\iota=1 \Longleftrightarrow$ comontonicity and $\iota=0 \Longleftrightarrow$ independence:\\
The boundedness of $\iota$ follows from the upper and lower bounds of $\omega$ and $\gamma$. Then, by the definition of $\iota^U$ and $\iota^L$ it will always hold that both are bounded between -1 and 1. \\
First, $\iota=-1\Longleftrightarrow$ counter-monotonicity\\
\begin{proof}
    ($\Longleftarrow$)Take a bivariate counter monotonic random vector $(X,Y)$ then $C(u,v)=\max(u+v-1,0)$ and $\Bar{C}(u,v)=1-u-v+\max(u+v-1,0)$. It follows that for any $\alpha,\beta\in(0,1)$ the largest solution to  $\max(\alpha+\omega-1,0)=\alpha\beta$ is $\omega=1-\alpha(1-\beta)$. Equivalently, from $\gamma$ being the smallest solution to $1-\alpha-\gamma+\max(\alpha+\gamma-1,0)=(1-\alpha)(1-\beta)$ follows that $\gamma=\beta(1-\alpha)$. Then, it follows that $\iota^L=-1,\iota^U=-1$.\\
    ($\Longrightarrow$) Take a bivariate random vector $(X,Y)$ and assume it is not counter monotonic. Then $C(u,v)>\max(u+v-1,0)$. By dependence consistency we then have that $\omega<1-\alpha(1-\beta)$ and $\gamma>\beta(1-\alpha)$. Then it follows that $\iota^L\neq-1,\iota^U\neq-1$.
\end{proof}
\noindent Second, $\iota=1\Longleftrightarrow$ comonotonicity\\
\begin{proof}
    ($\Longleftarrow$)Take a bivariate comonotonic random vector $(X,Y)$ then $C(u,v)=\min(u,v)$ and $\Bar{C}(u,v)=1-u-v+\min(u,v)$. It follows that for any $\alpha,\beta\in(0,1)$ the largest solution to $\min(\alpha,\omega)=\alpha\beta$ is $\omega=\alpha\beta$. Equivalently, from $\gamma$ being the smallest solution to $1-\alpha-\gamma+\min(\alpha,\gamma)=(1-\alpha)(1-\beta)$ it follows that $\gamma=\alpha+\beta-\alpha\beta$. Then it follows that $\iota^L=1,\iota^u=1$.
    ($\Longrightarrow$) Take a bivariate random vector $(X,Y)$ and assume it is not comonotonic. Then $C(u,v)<\min(u,v)$. By dependence consistency we then have that $\omega>\alpha\beta$ and $\gamma<\alpha+\beta-\alpha\beta$. Then it follows that $\iota^L\neq 1,\iota^U\neq 1$.
\end{proof}
\noindent Third, $\iota=0\Longleftrightarrow $ independence
\begin{proof}
      ($\Longleftarrow$)Take a bivariate independent random vector $(X,Y)$ then $C(u,v)=uv$ and $\Bar{C}(u,v)=1-u-v+uv$. It follows that for any $\alpha,\beta\in(0,1)$ the largest solution to $\alpha\omega=\alpha\beta$ is $\omega=\beta$. Equivalently, from $\gamma$ being the smallest solution to $1-\alpha-\gamma+\alpha\gamma=(1-\alpha)(1-\beta)$ it follows that $\gamma=\beta$. Then it follows that $\iota^L=0,\iota^u=0$.
    ($\Longrightarrow$) Take a bivariate random vector $(X,Y)$ and assume it is not independent. Then $C(u,v)\neq uv$. By dependence consistency we then have that $\omega\neq \beta$ and $\gamma\neq \beta$. Then it follows that $\iota^L\neq 0,\iota^U\neq 0$.
\end{proof}

\noindent 3. Invariant under strictly increasing functions of $X$ and $Y$: \\

\begin{proof}
    It is a basic result in the copula literature that copulas are invariant under strictly increasing functions of the marginals (see Theorem 2.4.3 in \cite{nelsen2007introduction} or Theorem 2.4.7 in \cite{hofert2019elements}). From this it easily follows that then $\omega,\gamma$ and hence $\iota^U$ and $\iota^L$ are unchanged. 
\end{proof}
\noindent 4. Risk asymmetry: $\iota_{\alpha,\beta}(Y\mid X)= \iota_{\alpha,\beta}(X\mid Y)$ $\Longleftrightarrow$ the copula is symmetric $C(u,v)=C(v,u)$. 
\begin{proof}
    ($\Longleftarrow$) Assume the copula is symmetric. Hence, $C(u,v)=C(v,u)$ for all $v,u\in(0,1)$. Then, it follows that $\omega$ and $\gamma$ are unchanged if we change the order of $X$ and $Y$. Therefore, $\iota_{\alpha,\beta}(Y\mid X)= \iota_{\alpha,\beta}(X\mid Y)$ for all $\alpha,\beta\in(0,1)$.\\
    ($\Longrightarrow$) Assume  $\iota_{\alpha,\beta}(Y\mid X)= \iota_{\alpha,\beta}(X\mid Y)$ for all $\alpha,\beta\in(0,1)$ from the definition of $\omega$ and $\gamma$ it then follows that $C(\alpha,\omega)=C(\omega,\alpha)$ and $\Bar{C}(\alpha,\gamma)=\Bar{C}(\gamma,\alpha)$ for all $\alpha,\beta\in(0,1)$. Therefore the copula is symmetric. 
\end{proof}
\noindent 5. Tail asymmetry:  $\iota^L_{\alpha,\beta}(Y\mid X)=\iota^U_{1-\alpha,1-\beta}(Y\mid X)$ $\Longleftrightarrow$ the copula is radially symmetric $C(u,v)=\Bar{C}(1-u,1-v)$.
\begin{proof}
    ($\Longleftarrow$) Assume the copula is radially symmetric. This implies that $C(u,v)=\Tilde{C}(u,v)$ for all $u,v\in(0,1)$ with $\Tilde{C}(u,v)$ the survival copula . By definition $\Tilde{C}(u,v)=u+v-1+C(1u,1-v)$. Then, it follows that $C(u,v)=\Bar{C}(1-u,1-v)$ for all $u,v\in(0,1)$ (see also \cite{gijbels2021specification} for this result). From the definitions $\ref{def:CovCopula},\ref{def:Covcopula2}$ it then follows that for all $\alpha,\beta\in(0,1)$ $\omega(\alpha,\beta,C)=1-\gamma(1-\alpha,1-\beta,C)$. From the definitions of $\iota^L_{\alpha,\beta}(Y\mid X)$ and $\iota^U_{1-\alpha,1-\beta}(Y\mid X)$ it then follows that $\iota^L_{\alpha,\beta}(Y\mid X)=\iota^U_{1-\alpha,1-\beta}(Y\mid X)$.\\
    ($\Longrightarrow$) Assume $\iota^L_{\alpha,\beta}(Y\mid X)=\iota^U_{1-\alpha,1-\beta}(Y\mid X)$. From the definitions of $\iota^L_{\alpha,\beta}(Y\mid X)$ and $\iota^U_{1-\alpha,1-\beta}(Y\mid X)$ it then follows that $\omega(\alpha,\beta,C)=1-\gamma(1-\alpha,1-\beta,C)$. From definitions \ref{def:CovCopula}, \ref{def:Covcopula2} it then follows that $C(\alpha,\omega)=\Bar{C}(1-\alpha,1-\gamma)$ which is equivalent to radial symmetry of the copula. 
\end{proof}
\noindent 6. Locality: $\iota_{\alpha,\beta}(Y\mid X)$ measures the dependence between $X$ and $Y$ at the quantiles given by $\alpha,\beta$:\\
\noindent This can be seen by the fact that $\omega,\gamma$ are both functions of the set of quantiles one chooses. Therefore, the properties of the copula and hence the dependence structure in the local region of the solutions $\omega$ and $\gamma$ determine their value.\\
\noindent 7. Consistency and asymptotic normality:\\
\noindent First, we state some definitions. \\
Let
\begin{align}\label{eq:empcop}
    \hat{C}_n(u,v)=\frac{1}{n}\sum^n_{i=1}\mathbb{I}(\hat{U}_i\leq u,\hat{V}_i\leq v)
\end{align}
be the empirical copula with $\hat{U}_i=\hat{F}(X_i)$, $\hat{F}(X_i)$ empirical cdf of $X$, 
$\hat{V}_i=\hat{G}(Y_i)$, $\hat{G}(Y_i)$ the empirical cdf of $Y$. Let
\begin{align}
    \hat{C}^\beta_n(u,v)=\frac{1}{n}\displaystyle\sum^n_{i=1}F_{n,R^{(n)}_{i,1}}(u)F_{n,R^{(n)}_{i,2}}(v) 
\end{align}
be the empirical beta copula with  $F_{n,r}(u)=\displaystyle\sum^n_{s=r}\begin{pmatrix}
    n \\s
\end{pmatrix} u^s(1-u)^{n-s}$. According to \cite{SEGERS201735}  $F_{n,r}(u)=\displaystyle\sum^n_{s=r}\begin{pmatrix}
    n \\s
\end{pmatrix} u^s(1-u)^{n-s}=\sum^n_{s=0}\mathbb{I}(r\leq s)\begin{pmatrix}
    n \\s
\end{pmatrix}u^s(1-u)^{n-s}$ which means that the empirical beta copula is still an empirical copula because it smoothes the empirical copula with a Bernstein polynomial. Hence, it is also a Bernstein copula \cite{janssen2012largebernstein} with degree equal to $n$. In order for our results to hold we put assumptions on the copula. Following \cite{segers2012asymptotics} and \cite{SEGERS201735} we assume the copula $C(u,v)$ has continuous partial derivatives except possibly on the boundary points (0,0),(0,1),(1,0) and (1,1). Furthermore, assume all observations $(X_1,Y_1),\dots,(X_n,Y_n)$ are i.i.d. \\
First, we will establish uniform and point-wise consistency of $\hat{\iota}_{\alpha,\beta}^L(Y\mid X)$ and $\hat{\iota}_{\alpha,\beta}^U(Y\mid X)$:
\begin{proof}
    According to \cite{SEGERS201735} the empirical beta copula is uniformly consistent whenever the empirical copula is. Since \cite{deheuvels1979fonction} has proven uniform consistency of the empirical copula it follows that:
    \begin{align*}
        \sup_{[u,v]\in[0,1]^2}\mid \hat{C}^\beta_n(u,v)-C(u,v)\mid \xrightarrow{a.s.}0, n\xrightarrow[]{}\infty.
    \end{align*}
    From this point-wise consistency in the sense that for all $(u,v)\in[0,1]^d$ $\hat{C}^\beta_n(u,v)\xrightarrow{a.s.}C(u,v)$ immediately follows (see \cite{van2000asymptotic}). Assuming $\gamma(\alpha,\beta,C)$ and $\omega(\alpha,\beta,C)$ are continuous functionals of $C$ for all $\alpha,\beta\in(0,1)$ by the extended continuous mapping Theorem (\cite{van2000asymptotic} Theorem 18.11) it follows that $\hat{\gamma}$ and $\hat{\omega}$ are uniform and point-wise consistent estimators of $\gamma$ and $\omega$. Then, again by the extended continuous mapping Theorem it follows that $\hat{\iota}_{\alpha,\beta}^L(Y\mid X)$ and $\hat{\iota}_{\alpha,\beta}^U(Y\mid X)$ are uniform and point-wise consistent estimators. 
\end{proof}
\noindent First, we define the concept of weak convergence following \cite{vaart1996empirical,van2000asymptotic}. This is an extension of the Central Limit Theorem to function spaces. For example if $\sqrt{n}(\hat{F}(t)-F(t))$ converges weakly in $D[-\infty,\infty]$ to a Gaussian process $G_F$ or $\sqrt{n}(\hat{F}(t)-F(t))\rightsquigarrow G_F$ it means the whole empirical CDF $\hat{F}$ will fluctuate in the set of all CDFs $D[-\infty,\infty]$ in a Gaussian manner around the true CDF $F$ as the sample size grows. This result, called Donskers Theorem, is also what the Kolmogorov Smirnov test is based on. Since it considers the convergence of the entire function in a "uniform" manner it is much stronger than the traditional point-wise asymptotic normality result of $\hat{F}$ which states that for all $t\in\Real$ $\sqrt{n}(\hat{F}(t)-F(t))\rightsquigarrow\mathcal{N}(0,F(t)(1-F(t))$ as $n\xrightarrow{}\infty$.\\
Next, we establish asymptotic normality of of $\hat{\iota}_{\alpha,\beta}^L(Y\mid X)$ and $\hat{\iota}_{\alpha,\beta}^U(Y\mid X)$ in a point-wise and uniform sense:
\begin{proof}
    Under the given assumptions on $C$ and assuming that $\sqrt (\hat{C}_n-C)\rightsquigarrow G$ where\begin{align*}
        G(u,v)=B(u,v)-\partial_uC(u,v)B(u,1)-\partial_vC(u,v)B(1,v)
    \end{align*} with $B(u,v)$ a Brownian Bridge on $[0,1]^2$ and \begin{align*}
        \cov(B(u,v),B(s,t))=C(\min(u,s),\min(v,t))-C(u,v)C(s,t)
    \end{align*} for $u\leq s,v\leq t$, \cite{SEGERS201735} proves that $\sqrt{n}(\hat{C}^\beta_n(u,v)-C(u,v))\rightsquigarrow G$ in the space of bounded functions $\ell^\infty([0,1]^2)$. Therefore, it follows also that for any $(u,v)\in[0,1]^2$ $\sqrt{n}(\hat{C}^\beta_n(u,v)-C(u,v))\rightsquigarrow \mathcal{N}(0,\sigma^2(u,v))$ with \begin{align*}
        \sigma^2(u,v)=C(u,v)(1-C(u,v))+u(1-u)\partial_u C(u,v)^2+v(1-v)\partial_v C(u,v)^2-2(1-u)C(u,v)\partial_u C(u,v)\\-2(1-v)C(u,v)\partial_v C(u,v)+2\partial_u C(u,v)\partial_v C(u,v)(C(u,v)-uv)
    \end{align*}
    derived in \cite{janssen2012largebernstein}. Assuming $\gamma$ and $\omega$ are Hadamard differentiable functions of $C$ by the Functional Delta Theorem \cite{van2000asymptotic} it follows that $\hat{\gamma}$ and $\hat{\omega}$ converge weakly and are point-wise asymptotically normal too. Now, because $\iota^L$ and $\iota^H$ are not Hadamard differentiable at $C(u,v)=uv$, assume that $C(u,v)\neq uv$. Then by again applying the Functional Delta Theorem we obtain weak convergence and hence point-wise asymptotic normality of $\hat{\iota}_{\alpha,\beta}^L(Y\mid X)$ and $\hat{\iota}_{\alpha,\beta}^U(Y\mid X)$.
\end{proof}
\noindent While the assumptions on the copula are as minimal as possible to obtain weak convergence to a Gaussian process $G$ with continuous trajectories they still leave out cases like the counter- and comonotonic copulas. Hence, when the copula is close to these boundary cases along with independence one should be wary of the $\iota$ estimates. According to \cite{segers2012asymptotics} the weak convergence could still hold in a more restricted space of functions like the space of càdlàg functions on $[0,1]^2$. point-wise consistency will still be preserved though since this requires weaker assumptions. Also, one could examine if the $\omega,\gamma$ estimates are close to $\beta$ since these do retain all nice properties. Furthermore, according to \cite{SEGERS201735} the results will still hold if the i.i.d. assumption is replaced by one of strict stationarity and weak dependence. These kinds of results are explored in \cite{bucher2013empirical}. 
\newpage
\subsubsection{Proofs of properties of $\delta$}\label{proof:delta}
We aim to establish the following properties: 
\begin{enumerate}
 \item Dependence consistency. 
    \item If $F_Y^{-1}$ is strictly increasing: $\delta\in[-1,1]$, $\delta=-1\Longleftrightarrow$ counter-monotonicity, $\delta=1 \Longleftrightarrow$ comontonicity and $\delta=0 \Longleftrightarrow$ independence. 
    \item Invariant under strictly increasing functions of $X$ and strictly increasing linear functions of $Y$.
    \item  Risk asymmetry: $\delta_{\alpha,\beta}(Y\mid X)= \delta_{\alpha,\beta}(X\mid Y)$ $\Longleftarrow$ the copula is symmetric $C(u,v)=C(v,u)$ and $X$ is a linear function of $Y$. 
    \item Tail asymmetry: $\delta^L_{\alpha,\beta}(Y\mid X)=\delta^U_{1-\alpha,1-\beta}(Y\mid X)$  $\Longleftarrow$ the copula is radially symmetric $C(u,v)=\Bar{C}(1-u,1-v)$ and $Y$ is symmetric.
    \item Locality: $\delta_{\alpha,\beta}(Y\mid X)$ measures the dependence between $X$ and $Y$ at the quantiles given by $\alpha,\beta$.
\item Consistency. 
    \end{enumerate}
    \noindent 1. Dependence consistency:
    \begin{proof}
        Given dependence consistency of $\omega$ and $\gamma$ shown in section \ref{proof:iota} we get that $\omega'\leq \omega$ and $\gamma'\geq \gamma$ if $C'(u,v)\geq C(u,v)$ for all $u,v\in (0,1)$. Since CDFs are by definition nondecreasing quantile functions are nondecreasing too. Therefore $F_{Y'}^{-1}(\omega')\leq F_{Y}^{-1}(\omega)$ and $F_{Y'}^{-1}(\gamma')\geq F_Y^{-1}(\gamma)$. From this it follows that $\delta^L_{\alpha,\beta}(Y'\mid X')\geq\delta^L_{\alpha,\beta}(Y\mid X)$ and $\delta^U_{\alpha,\beta}(Y'\mid X')\geq\delta^U_{\alpha,\beta}(Y\mid X)$.
    \end{proof}
    \noindent Alternatively, \cite{dhaene2022systemic} proves the dependence consistency of the $\DCov$. Because any $\delta$ is the $\DCov$ multiplied with a constant it follows that $\delta$ is dependence consistent as well.\\ 
    
    \noindent 2.  If $F_Y^{-1}$ is strictly increasing: $\delta\in[-1,1]$, $\delta=-1\Longleftrightarrow$ counter-monotonicity, $\delta=1 \Longleftrightarrow$ comontonicity and $\delta=0 \Longleftrightarrow$ independence:\\
    \noindent
    \begin{proof}
        All of these results follow from the results on $\iota$ but now for the $\Longrightarrow$ implications we require a strictly monotonic quantile function of $Y$ to ensure that $F_Y^{-1}(p')>F_Y^{-1}(p)$ for $p'>p$. 
    \end{proof}
    \noindent
    3. Invariant under strictly increasing functions of $X$ and strictly increasing linear functions of $Y$:\\
    \noindent 
    \begin{proof}
        The result on invariance to strictly increasing functions of $X$ follows from reasoning similar to the proof for $\iota$. For the result on $Y$ consider $Y'=a+bY$ with $a\in\Real$ and $b>0$. Then $F_{Y'}^{-1}(q)=a+bF_Y^{-1}(q)$ for all $q\in(0,1)$. Hence, it follows that $\delta$ will be invariant. 
    \end{proof}
    \noindent 4.  Risk asymmetry: $\delta_{\alpha,\beta}(Y\mid X)= \delta_{\alpha,\beta}(X\mid Y)$ $\Longleftarrow$ the copula is symmetric $C(u,v)=C(v,u)$ and  $X$ is a linear function of $Y$:
    \begin{proof}
        ($\Longleftarrow$) From the proofs in section \ref{proof:iota} it follows that $\omega,\gamma$ remain invariant. Because $X$ is a linear function of $Y$ we get that $F_{Y}^{-1}(q)=a+bF^{-1}_X(q)$ for all $q\in(0,1)$ and for $a,b\in\Real$. Then, we have that $\delta_{\alpha,\beta}(Y\mid X)= \delta_{\alpha,\beta}(X\mid Y)$.
    \end{proof}
    \noindent 5. Tail asymmetry $\delta^L_{\alpha,\beta}(Y\mid X)=\delta^U_{1-\alpha,1-\beta}(Y\mid X)$ $\Longleftarrow$ the copula is radially symmetric $C(u,v)=\Bar{C}(1-u,1-v)$ and $Y$ is symmetric.
    \begin{proof}
From the proof in section \ref{proof:iota} we know that a radially symmetric copula implies that $\omega(\alpha,\beta,C)=1-\gamma(1-\alpha,1-\beta,C)$ for all $\alpha,\beta\in(0,1)$. Without loss of generality we assume the density of $Y$ is symmetric around 0. Then, it follows that $F_Y^{-1}(q)=-F_Y^{-1}(1-q)$ for all $q\in(0,1)$. From the definition \ref{def:delta} it then follows that $\delta^L_{\alpha,\beta}(Y\mid X)=\delta^U_{1-\alpha,1-\beta}(Y\mid X)$  for all $\alpha,\beta\in(0,1)$. 
\end{proof}
\noindent 6. Locality: $\delta_{\alpha,\beta}(Y\mid X)$ measures the dependence between $X$ and $Y$ at the quantiles given by $\alpha,\beta$:\\
\noindent Similar reasoning as in section \ref{proof:iota} applies. However, $\delta$ also depends on the quantile function. Therefore, $\delta$ measures that dependence structure and the risk on $Y$ locally. \\
\noindent 7. Consistency:\\
\begin{proof}
    Given the uniform consistency of $\hat{\omega}$, $\hat{\gamma}$ and $\hat{F}_Y^{-1}$ (see \cite{van2000asymptotic}) applying the continuous mapping Theorem it follows that estimates of $\delta^L$ and $\delta^U$ are also uniformly consistent. 
\end{proof}
\noindent Asymptotic normality does not follow because for $\delta$ this would result in a ratio of asymptotically normal estimators which would, asymptotically, result in a ratio of normal distributions. In general, this is a heavy-tailed distribution \cite{hinkley1969ratio}.
\newpage
\subsection{Figures}\label{App:figs}
\subsubsection{Flash Crash}\label{App:FCfigs}
\begin{figure}[H]
     \centering
     \begin{subfigure}[t]{0.48\textwidth}
    \includegraphics[width=\textwidth]{Figures/FCplot99.pdf}
    \caption{5 minutes}
    \label{fig:FCplot5min}
\end{subfigure}
\begin{subfigure}[t]{0.48\textwidth}
    \includegraphics[width=\textwidth]{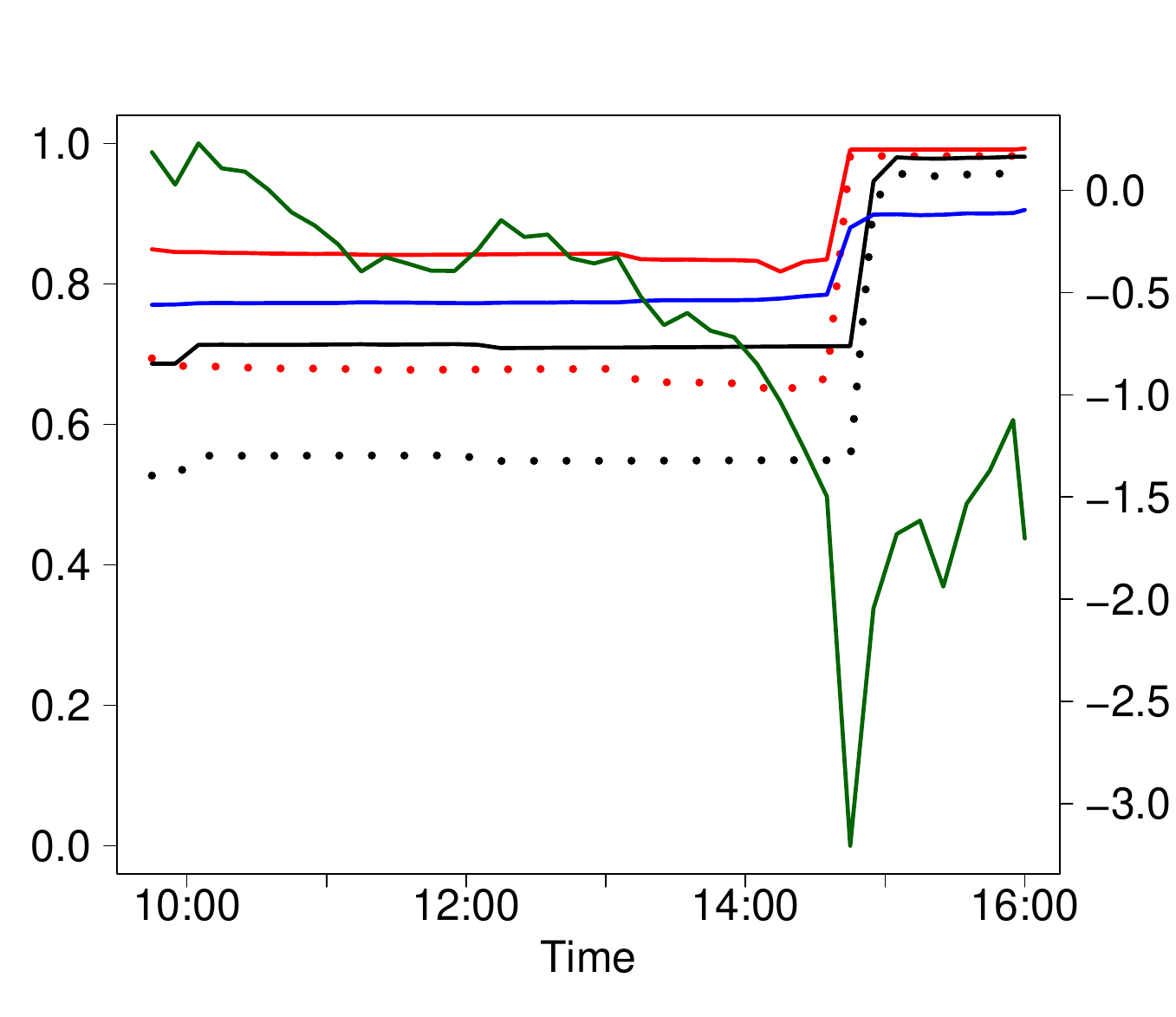}
    \caption{10 minutes }
    \label{fig:FCplot10min}
\end{subfigure}
\begin{subfigure}[t]{0.48\textwidth}
    \includegraphics[width=\textwidth]{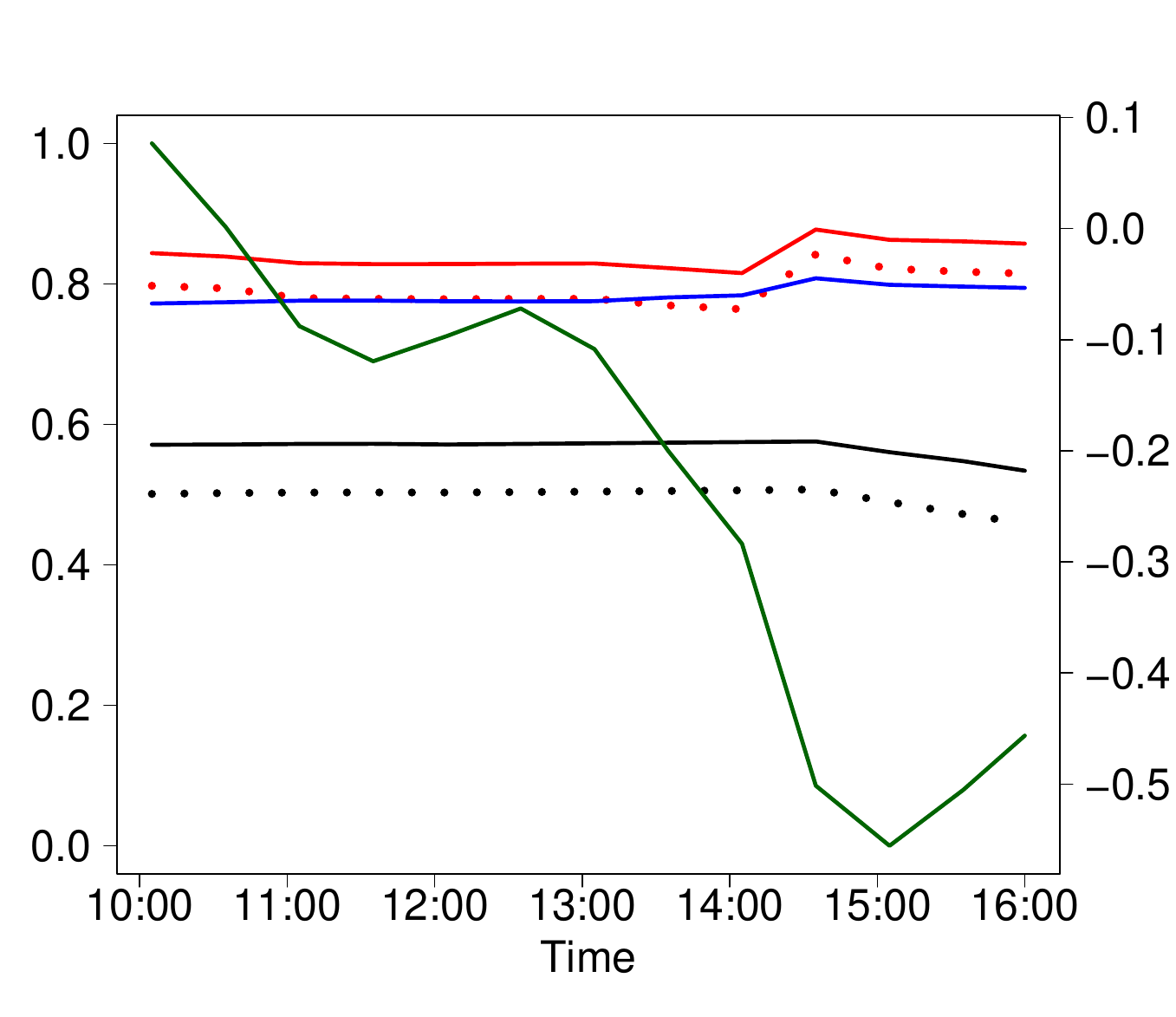}
    \caption{30 minutes}
    \label{fig:FCplot30min}
\end{subfigure}
\begin{subfigure}[t]{0.48\textwidth}
    \includegraphics[width=\textwidth]{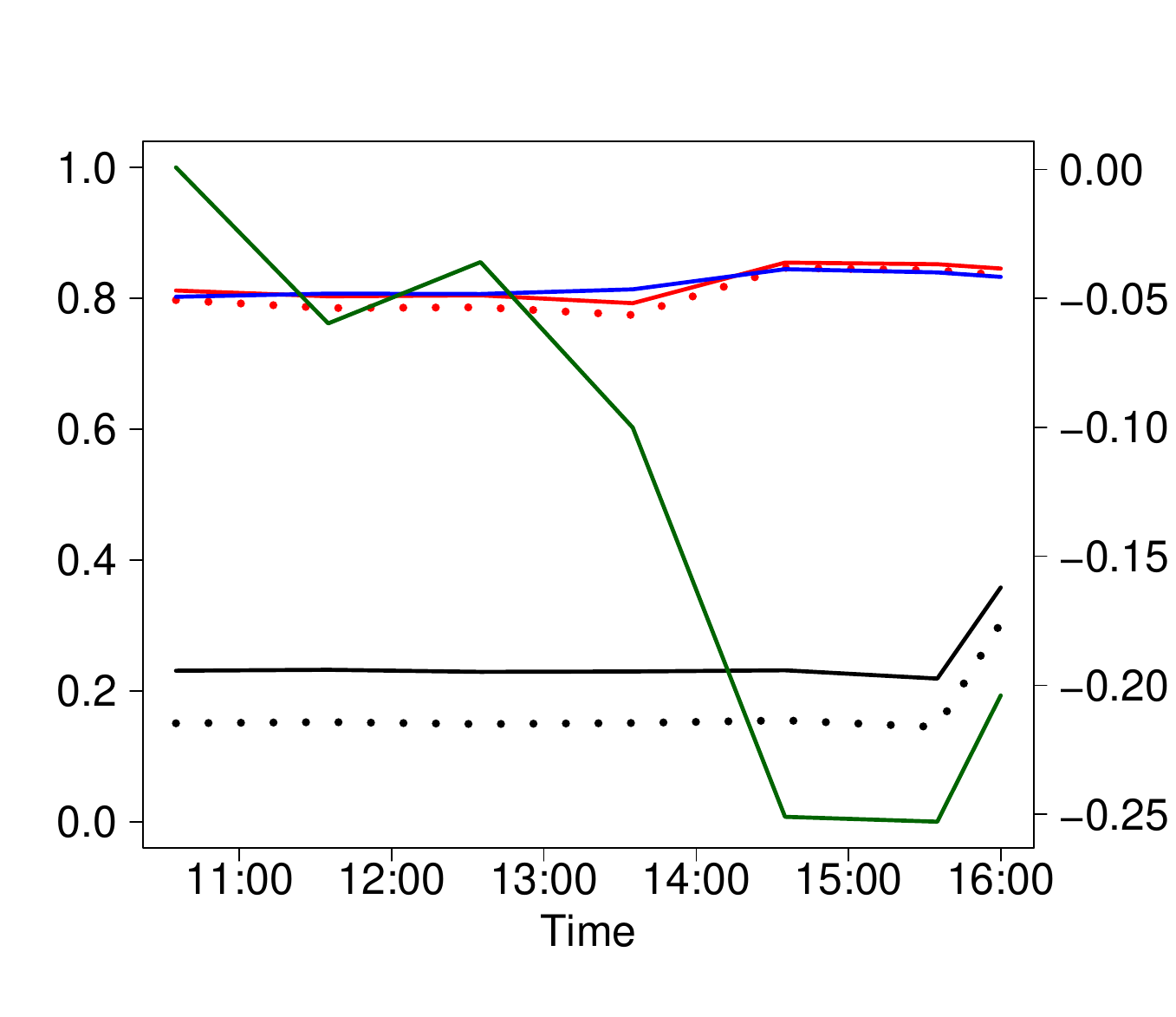}
    \caption{1 hour}
    \label{fig:FCplot1hr}
\end{subfigure}
\caption{Left: Estimates of the lower (red) and upper tail (black) $\iota$(solid) $\delta$(dotted) on the Flash Crash data compared with the fully Gaussian $\delta$(blue solid). Right: market index \% return (green line) . Lower: $\alpha=\beta=0.01$, upper: $\alpha=\beta=0.99$.}
\label{fig:FCtimecomp}
\end{figure}
\newpage
\subsubsection{September 15th}\label{App:Sept15}\
\begin{figure}[H]
     \centering
     \begin{subfigure}[t]{0.48\textwidth}
    \includegraphics[width=\textwidth]{Figures/GFC1plot99.pdf}
    \caption{Lower: $\alpha=\beta=0.01$, Upper: $\alpha=\beta=0.99$}
    \label{fig:GFC199}
\end{subfigure}
\begin{subfigure}[t]{0.48\textwidth}
    \includegraphics[width=\textwidth]{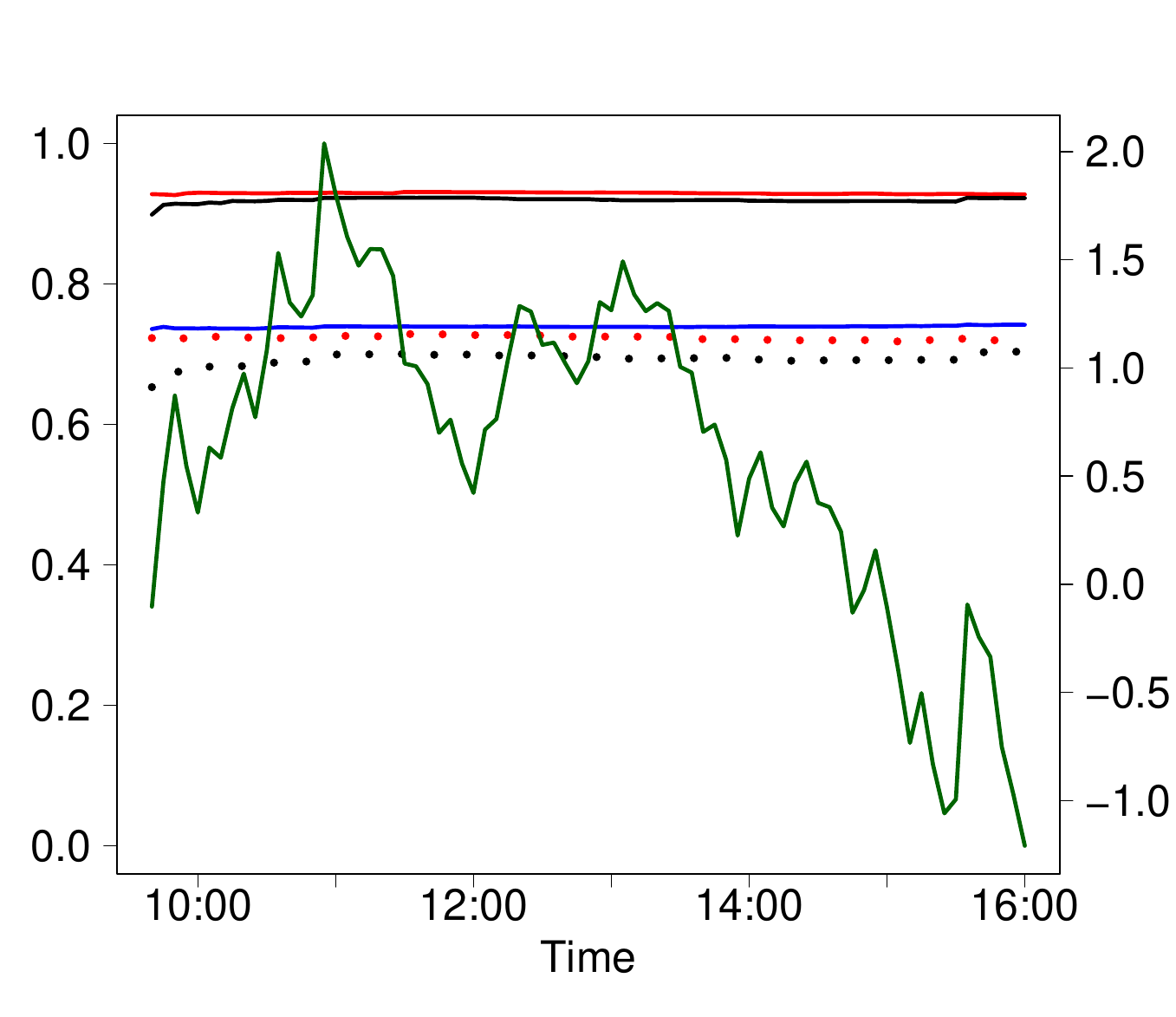}
    \caption{Lower: $\alpha=\beta=0.05$, Upper: $\alpha=\beta=0.95$ }
    \label{fig:GFC195}
\end{subfigure}
\begin{subfigure}[t]{0.48\textwidth}
    \includegraphics[width=\textwidth]{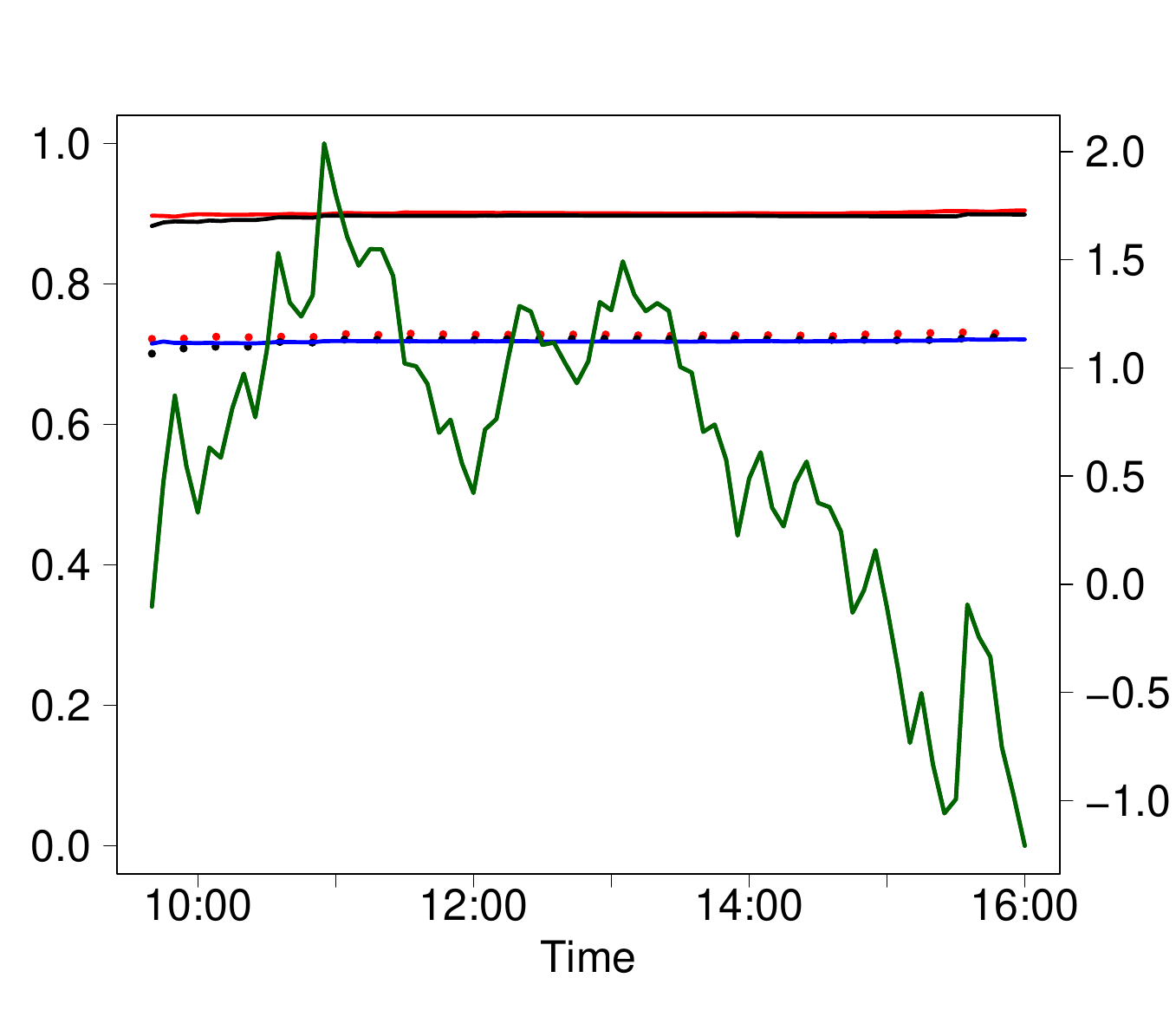}
    \caption{Lower: $\alpha=\beta=0.10$, Upper: $\alpha=\beta=0.90$}
    \label{fig:GFC190}
\end{subfigure}
\begin{subfigure}[t]{0.48\textwidth}
    \includegraphics[width=\textwidth]{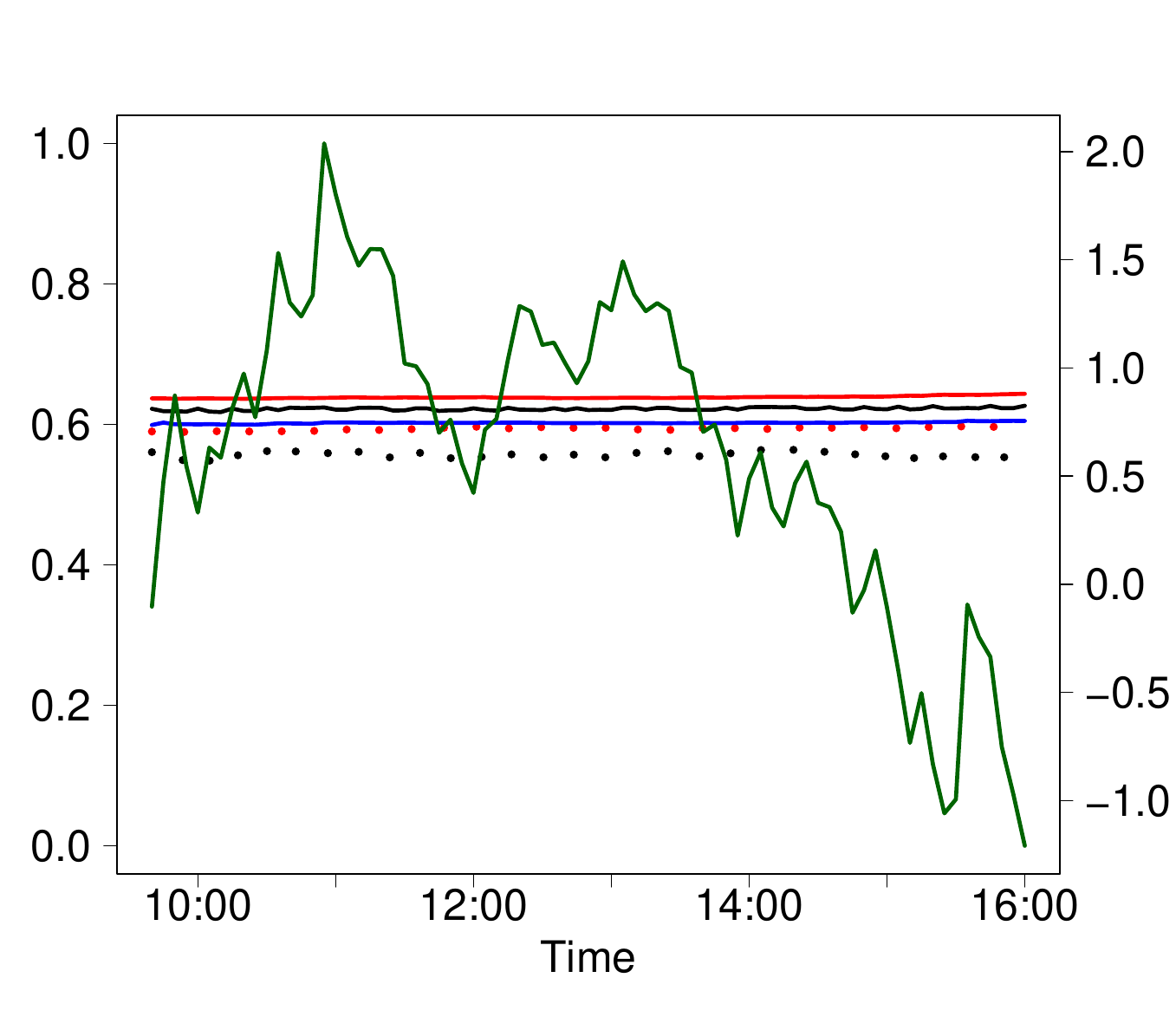}
    \caption{Lower: $\alpha=\beta=0.50$, Upper: $\alpha=\beta=0.50$}
    \label{fig:GFC150}
\end{subfigure}
\caption{Left: Estimates of the lower (red) and upper tail (black) $\iota$(solid) $\delta$(dotted) on September 15th 2008 data compared with the fully Gaussian $\delta$(blue solid). Right: market index \% return (green line). }
\label{fig:GFC1comp}
\end{figure}
\newpage
\subsubsection{September 16th}
\begin{figure}[H]
     \centering
     \begin{subfigure}[t]{0.48\textwidth}
    \includegraphics[width=\textwidth]{Figures/GFC2plot99.pdf}
    \caption{Lower: $\alpha=\beta=0.01$, Upper: $\alpha=\beta=0.99$}
    \label{fig:GFC299}
\end{subfigure}
\begin{subfigure}[t]{0.48\textwidth}
    \includegraphics[width=\textwidth]{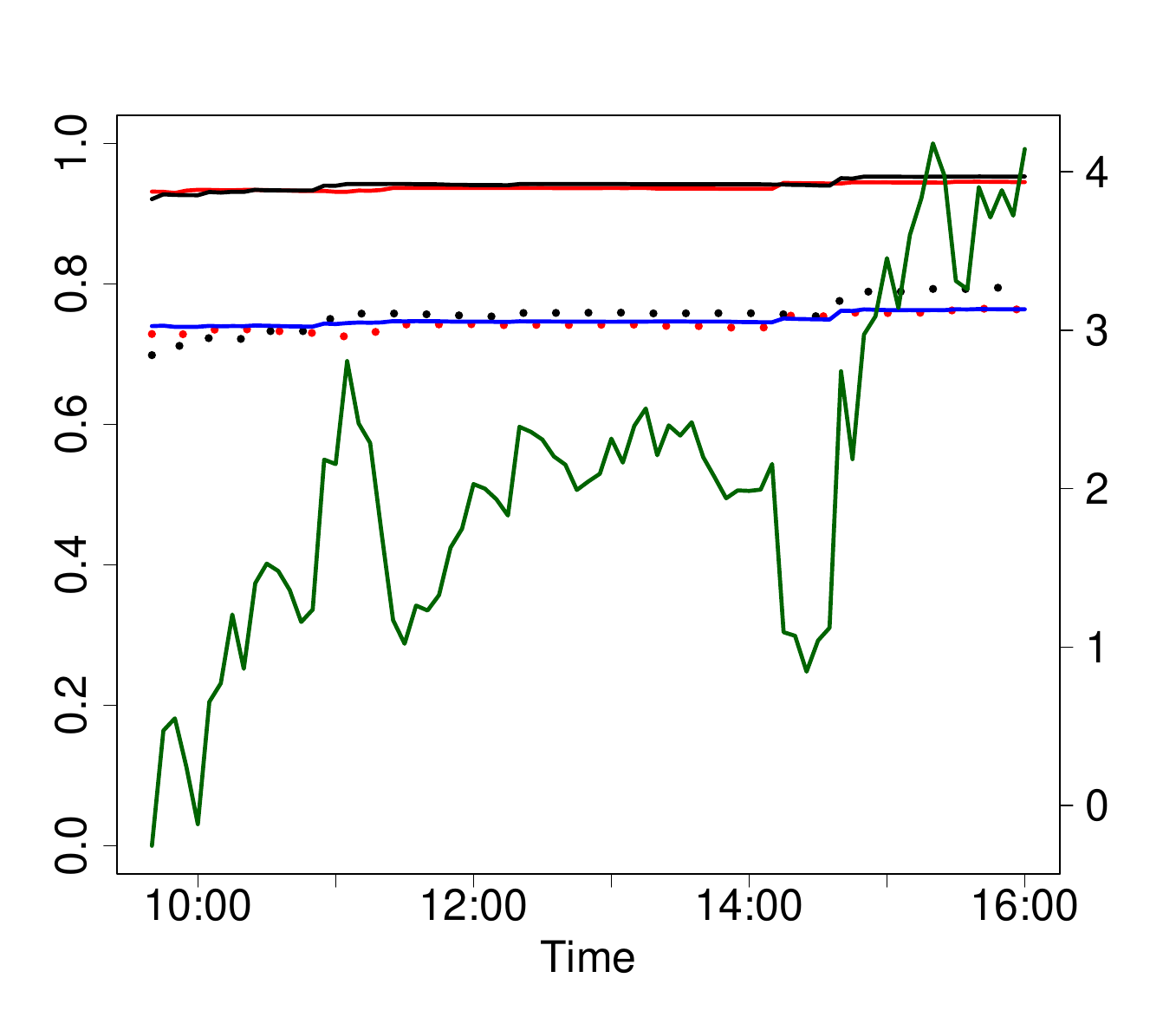}
    \caption{Lower: $\alpha=\beta=0.05$, Upper: $\alpha=\beta=0.95$ }
    \label{fig:GFC295}
\end{subfigure}
\begin{subfigure}[t]{0.48\textwidth}
    \includegraphics[width=\textwidth]{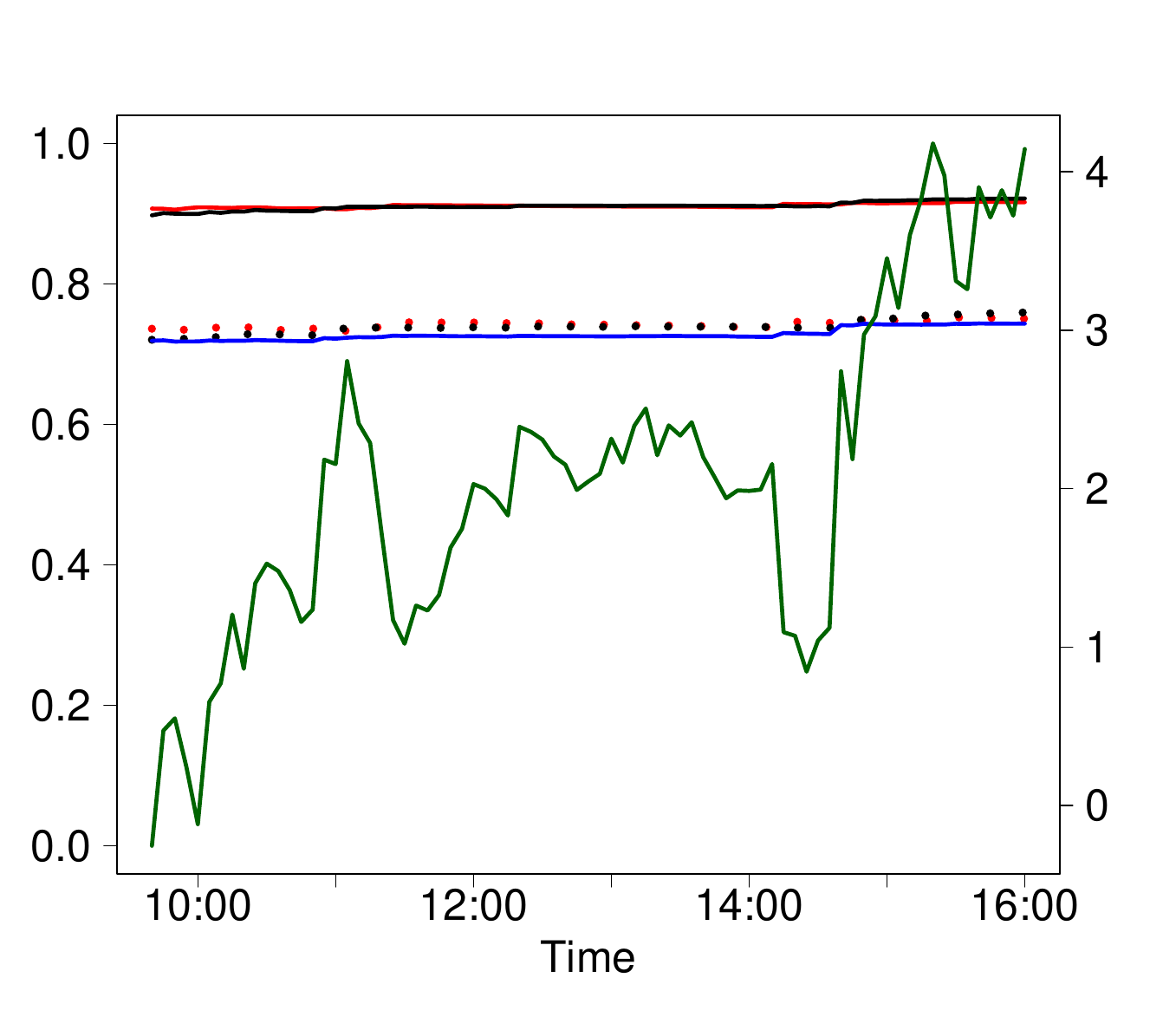}
    \caption{Lower: $\alpha=\beta=0.10$, Upper: $\alpha=\beta=0.90$}
    \label{fig:GFC290}
\end{subfigure}
\begin{subfigure}[t]{0.48\textwidth}
    \includegraphics[width=\textwidth]{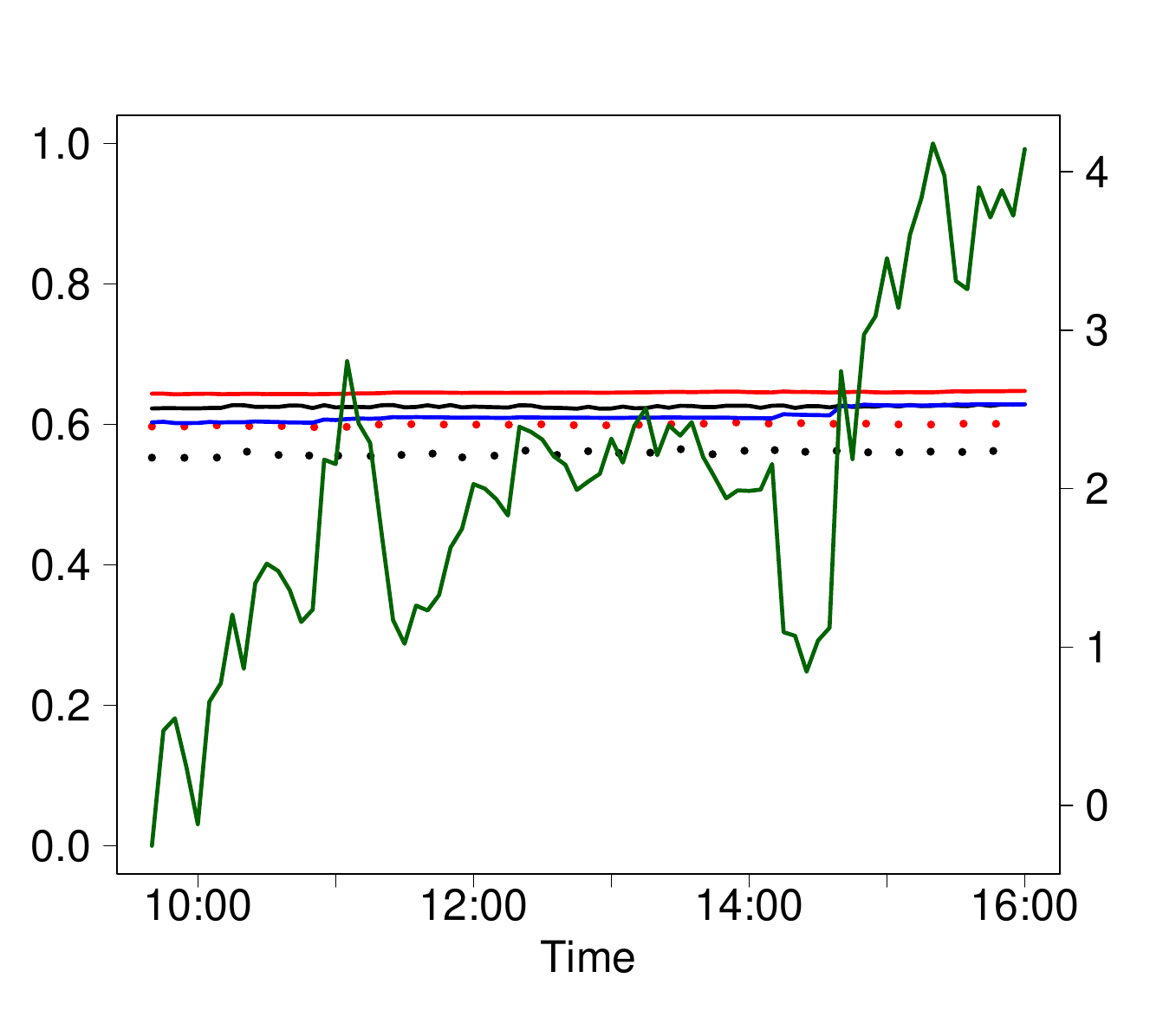}
    \caption{Lower: $\alpha=\beta=0.50$, Upper: $\alpha=\beta=0.50$}
    \label{fig:GFC250}
\end{subfigure}
\caption{Left: Estimates of the lower (red) and upper tail (black) $\iota$(solid) $\delta$(dotted) on September 16th 2008 data compared with the fully Gaussian $\delta$(blue solid). Right: market index \% return (green line). }
\label{fig:GFC2comp}
\end{figure}
\newpage
\subsubsection{September 29th}
\begin{figure}[H]
     \centering
     \begin{subfigure}[t]{0.48\textwidth}
    \includegraphics[width=\textwidth]{Figures/GFC3plot99.pdf}
    \caption{Lower: $\alpha=\beta=0.01$, Upper: $\alpha=\beta=0.99$}
    \label{fig:GFC399}
\end{subfigure}
\begin{subfigure}[t]{0.48\textwidth}
    \includegraphics[width=\textwidth]{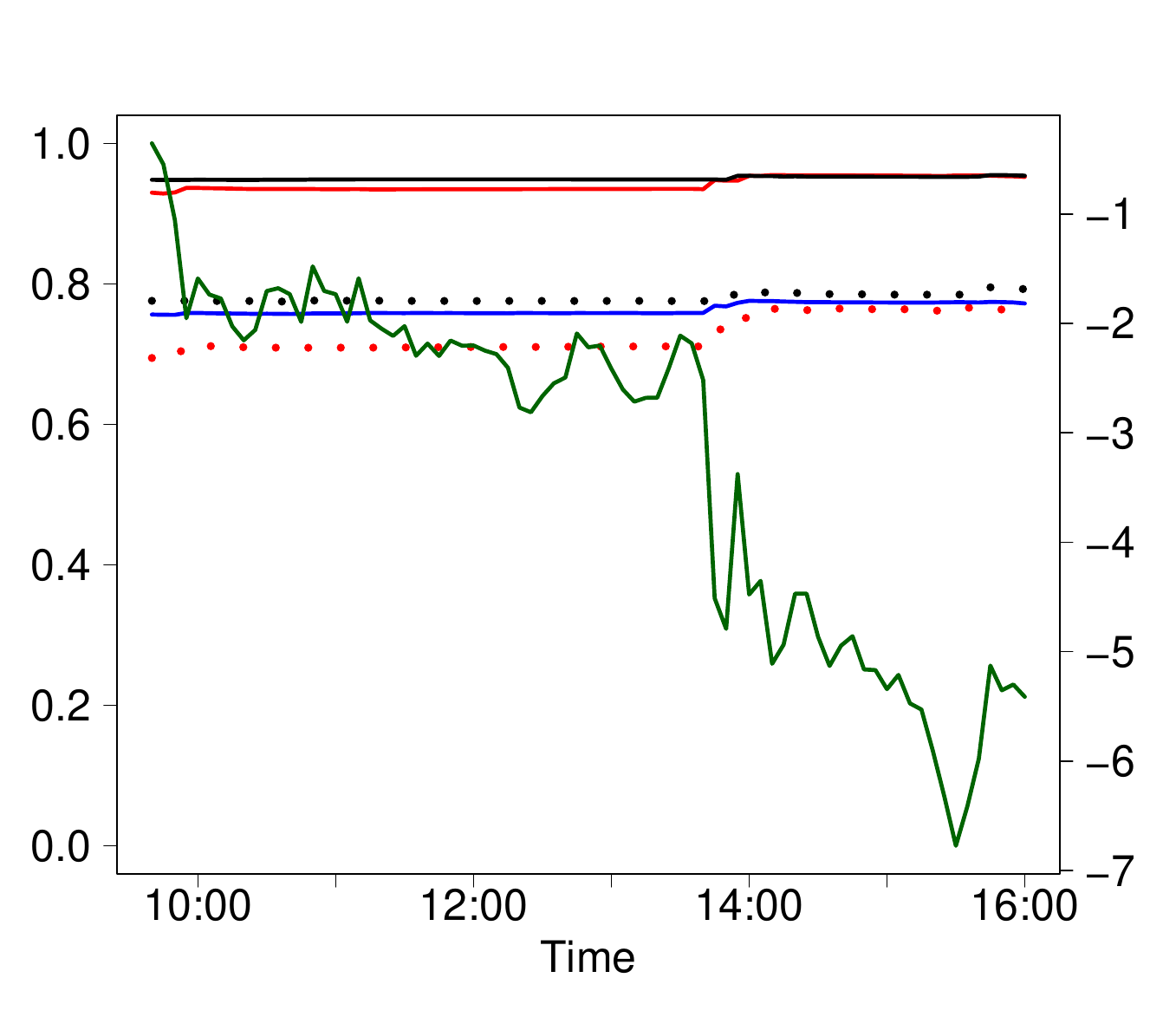}
    \caption{Lower: $\alpha=\beta=0.05$, Upper: $\alpha=\beta=0.95$ }
    \label{fig:GFC395}
\end{subfigure}
\begin{subfigure}[t]{0.48\textwidth}
    \includegraphics[width=\textwidth]{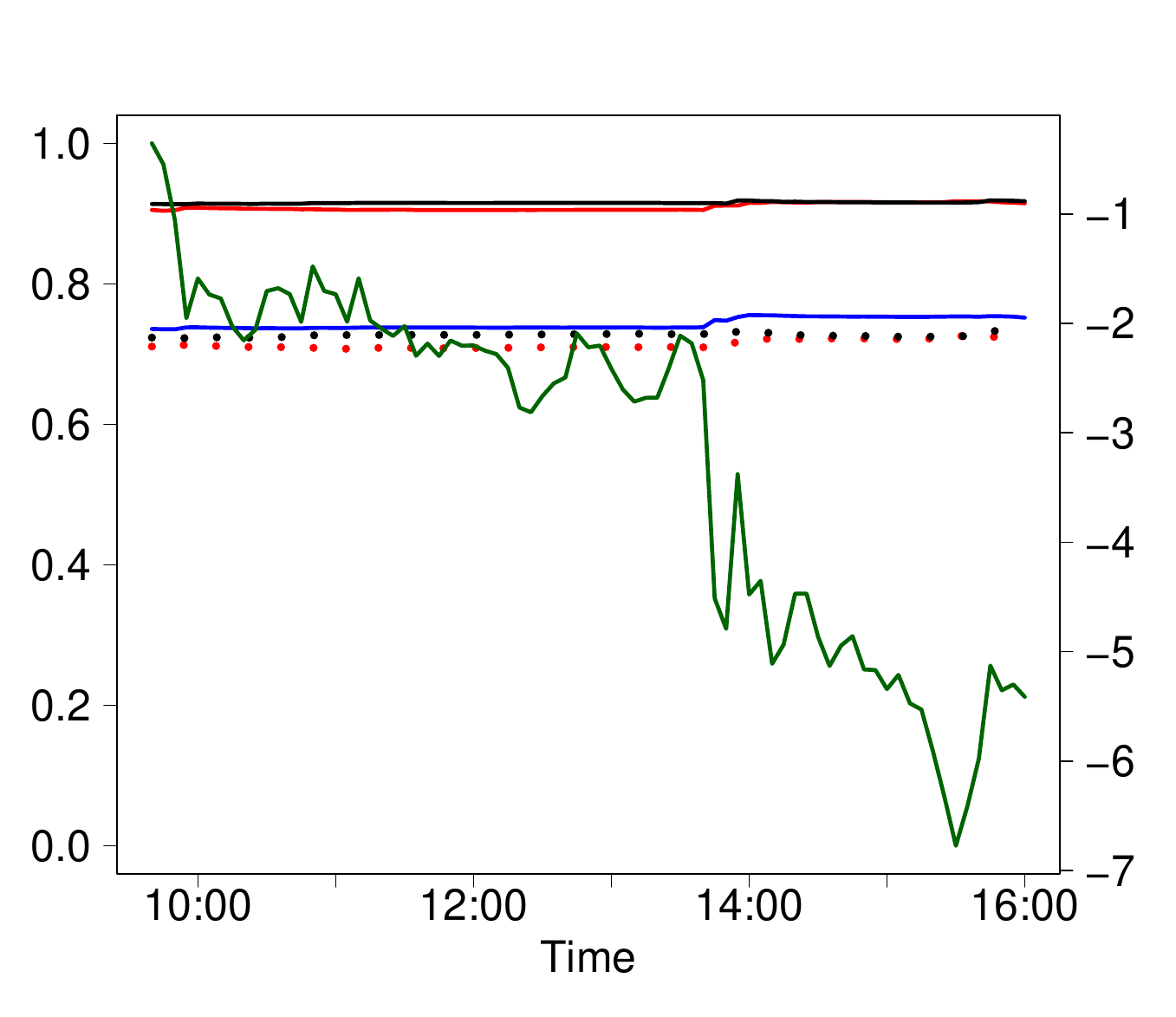}
    \caption{Lower: $\alpha=\beta=0.10$, Upper: $\alpha=\beta=0.90$}
    \label{fig:GFC390}
\end{subfigure}
\begin{subfigure}[t]{0.48\textwidth}
    \includegraphics[width=\textwidth]{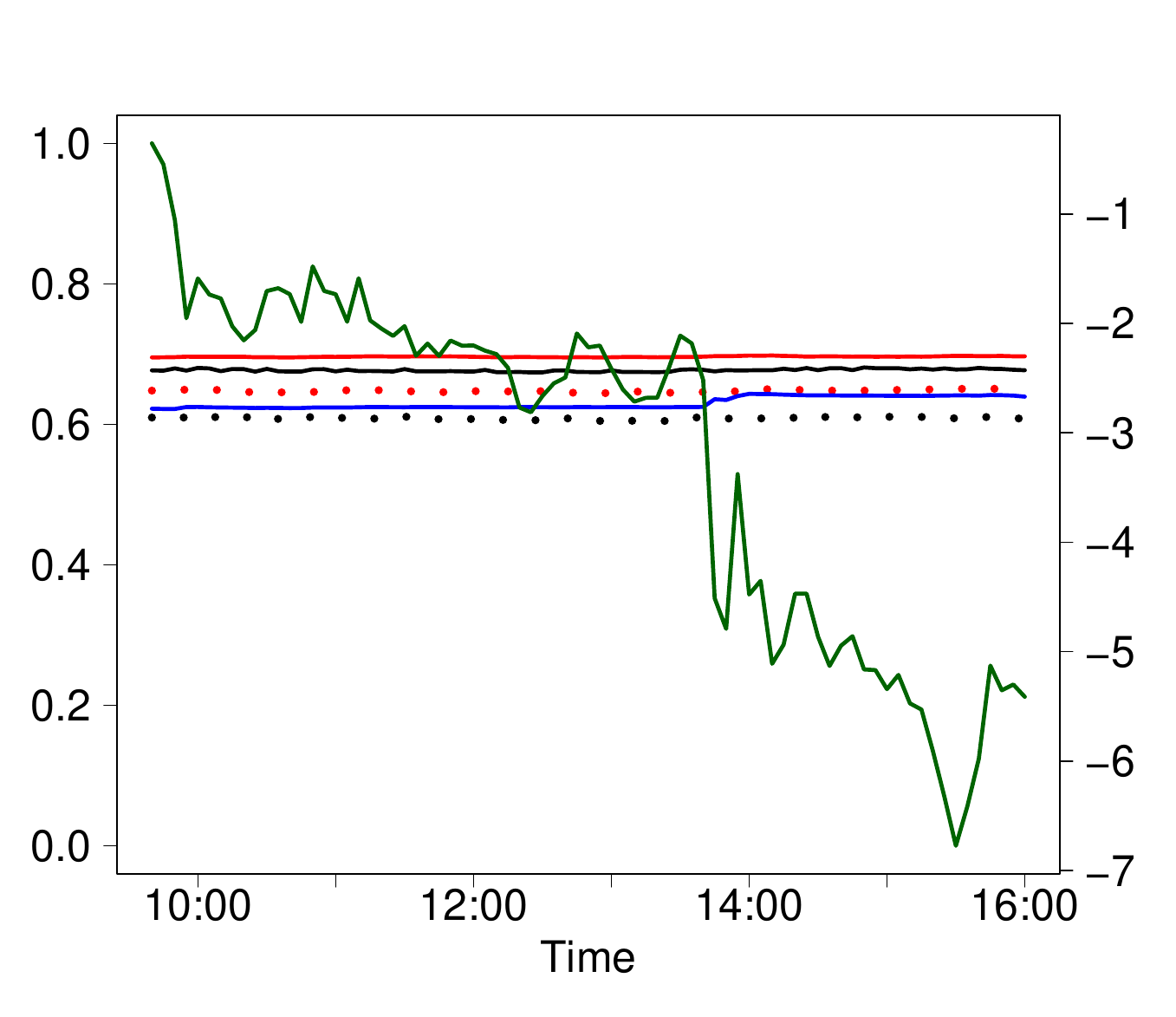}
    \caption{Lower: $\alpha=\beta=0.50$, Upper: $\alpha=\beta=0.50$}
    \label{fig:GFC350}
\end{subfigure}
\caption{Left: Estimates of the lower (red) and upper tail (black) $\iota$(solid) $\delta$(dotted) on September 29th 2008 data compared with the fully Gaussian $\delta$(blue solid). Right: market index \% return (green line).}
\label{fig:GFC3comp}
\end{figure}
\newpage
\subsubsection{October 3rd}
\begin{figure}[H]
     \centering
     \begin{subfigure}[t]{0.48\textwidth}
    \includegraphics[width=\textwidth]{Figures/GFC4plot99.pdf}
    \caption{Lower: $\alpha=\beta=0.01$, Upper: $\alpha=\beta=0.99$}
    \label{fig:GFC499}
\end{subfigure}
\begin{subfigure}[t]{0.48\textwidth}
    \includegraphics[width=\textwidth]{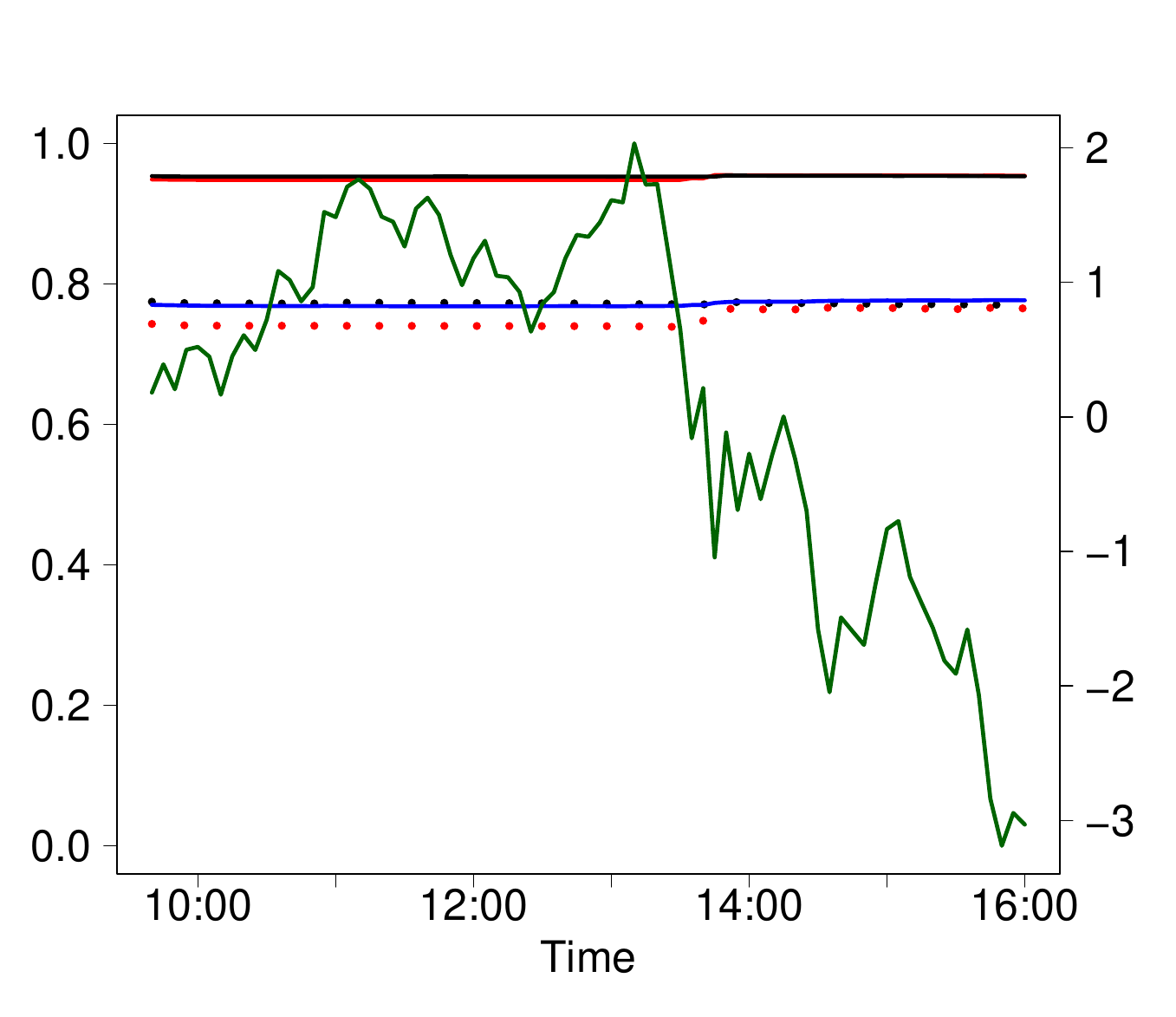}
    \caption{Lower: $\alpha=\beta=0.05$, Upper: $\alpha=\beta=0.95$ }
    \label{fig:GFC495}
\end{subfigure}
\begin{subfigure}[t]{0.48\textwidth}
    \includegraphics[width=\textwidth]{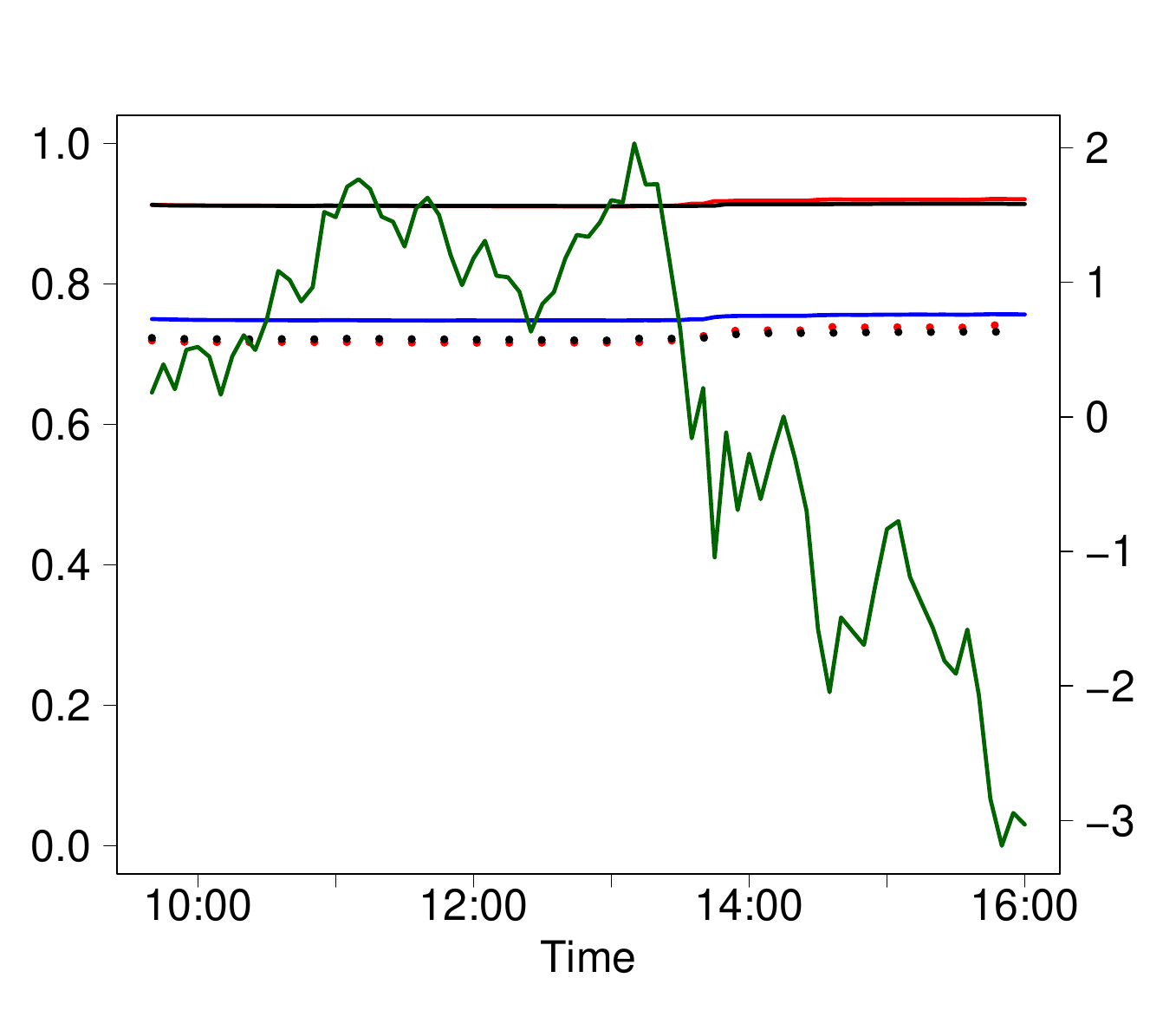}
    \caption{Lower: $\alpha=\beta=0.10$, Upper: $\alpha=\beta=0.90$}
    \label{fig:GFC490}
\end{subfigure}
\begin{subfigure}[t]{0.48\textwidth}
    \includegraphics[width=\textwidth]{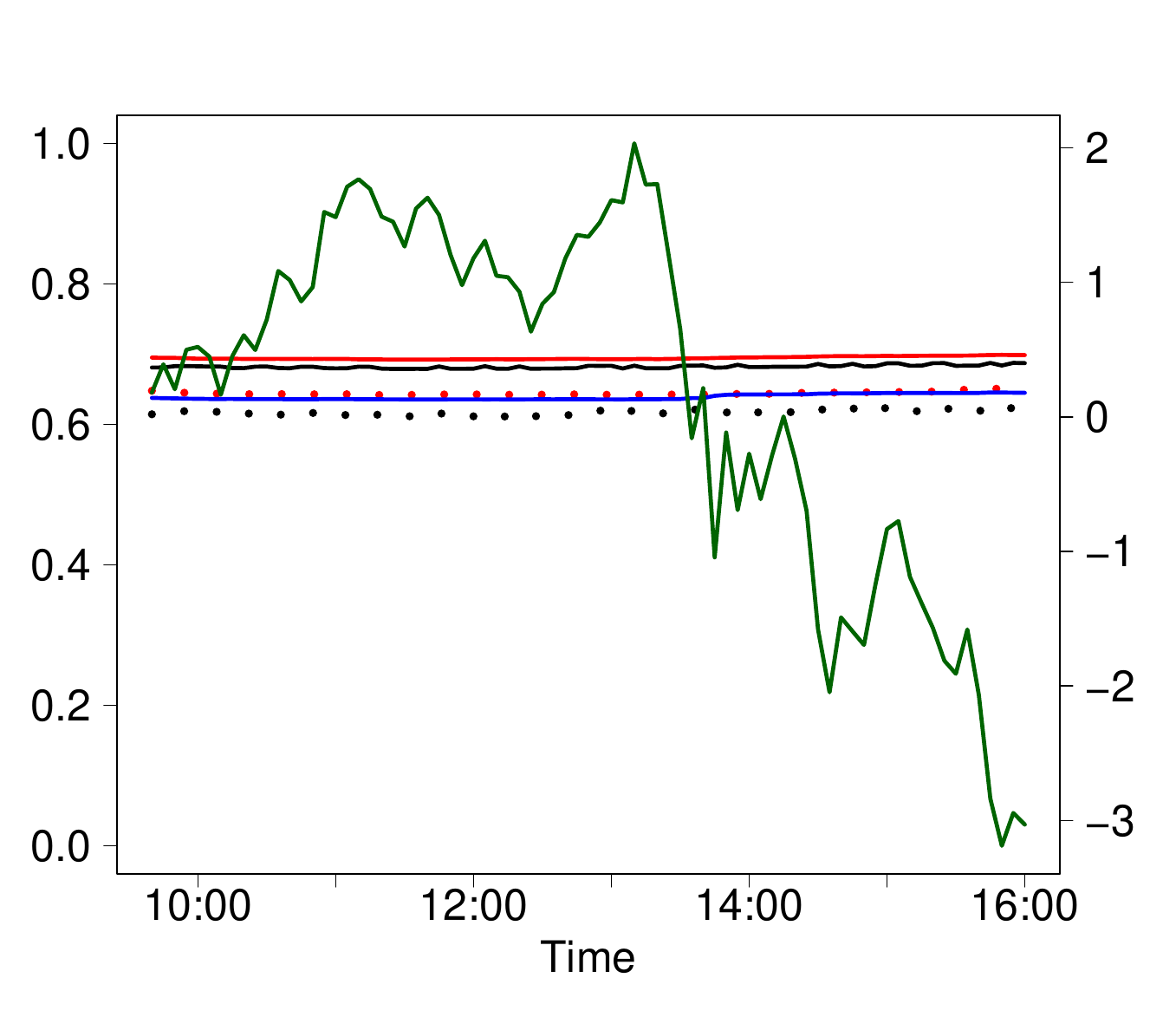}
    \caption{Lower: $\alpha=\beta=0.50$, Upper: $\alpha=\beta=0.50$}
    \label{fig:GFC450}
\end{subfigure}
\caption{Left: Estimates of the lower (red) and upper tail (black) $\iota$(solid) $\delta$(dotted) on October 3rd 2008 data compared with the fully Gaussian $\delta$(blue solid). Right: market index \% return (green line). }
\label{fig:GFC4comp}
\end{figure}
\newpage
\end{document}